\documentclass[superscriptaddress]{revtex4-2}

\usepackage{amssymb,amsthm,amsfonts,amstext,amsmath}
\usepackage{wrapfig}
\usepackage{caption}
\usepackage{subcaption}
\usepackage{float}
\usepackage{graphicx}
\usepackage{xcolor}
\usepackage{array}
\usepackage{physics}

\usepackage{url}
\usepackage[colorlinks=true,citecolor=cyan,urlcolor=magenta]{hyperref}
\usepackage{cleveref}

\usepackage{tcolorbox}
\usepackage{enumitem}

\def\be{\begin{equation}}
\def\ee{\end{equation}}
\def\bea{\begin{eqnarray}}
\def\eea{\end{eqnarray}}
\def\bma{\begin{mathletters}}
\def\ema{\end{mathletters}}

\def\q0{\underline{0}}

\def\sB{{\sf B}}

\def\CC{\mathcal{C}}
\def\sC{{\sf C}}
\def\DD{\mathcal{D}}
\def\id{{\mathbb I}}
\def\E{{\cal E}}
\def\G{{\cal G}}
\def\V{{\cal V}}

\def\sx{{\sigma_1}}
\def\sy{{\sigma_2}}
\def\sz{{\sigma_3}}

\def\R{\mathbb{R}}

\def\one{\leavevmode\hbox{\small1\normalsize\kern-.33em1}}

\def\bra#1{\langle#1|} \def\ket#1{|#1\rangle}
\def\braket#1#2{\langle#1|#2\rangle}

\def\id{{\mathbb I}}

\newtheorem{theorem}{Theorem}
\newtheorem{remark}[theorem]{Remark}

\newtheorem{lemma}[theorem]{Lemma}
\newtheorem{prop}[theorem]{Proposition}

\newtheorem{defin}[theorem]{Definition}

\begin{document}

\title{Tight and self-testing multipartite quantum Bell inequalities\\ from the renormalization group}

\author{Paolo Abiuso}
\affiliation{Institute for Quantum Optics and Quantum Information (IQOQI) Vienna\\ Austrian Academy of Sciences, Boltzmanngasse 3, Wien 1090, Austria}
\author{Julian Fischer} 
\affiliation{Institute for Theoretical Physics, Johannes Kepler University Linz\\ Altenberger Straße 69, Linz 4040, Austria}
\affiliation{Department of Engineering and Computer Sciences, Hamburg University of Applied Sciences\\ Berliner Tor 7, Hamburg 20099, Germany}
\author{Miguel Navascu\'es} 
\affiliation{Institute for Quantum Optics and Quantum Information (IQOQI) Vienna\\ Austrian Academy of Sciences, Boltzmanngasse 3, Wien 1090, Austria}

\begin{abstract}
In past work~\cite{connectors}, the concept of connectors was introduced: directed tensors with the property that any contraction thereof defines a multipartite quantum Bell inequality, i.e., a linear restriction on measurement probabilities that holds in any multipartite quantum experiment. In this paper we propose the notion of \emph{tight connectors}, which, if contracted according to some simple rules, result in tight quantum Bell inequalities. By construction, the new inequalities are saturated by tensor network states, whose structure mimics the corresponding network of connectors. Some tight connectors are furthermore ``fully self-testing'', which implies that the quantum Bell inequalities they generate can only be maximized with such a tensor network state and specific measurement operators (modulo local isometries). We provide large analytic families of tight, fully self-testing connectors that generate $N$-partite quantum Bell inequalities of correlator form for which the ratio between the maximum quantum and classical values increases exponentially with $N$.

\end{abstract}

\date{\today}

\maketitle

\onecolumngrid

\section{Introduction}
The correlations observed by distant parties measuring a joint quantum state do not, in general, admit a classical description. This is the gist of Bell’s theorem~\cite{Bell}, one of the deepest discoveries in quantum physics, which ten years ago led to several experimental disproofs of local realism~\cite{loophole_free1, loophole_free2, loophole_free3}. The quantum analog of Bell’s theorem is the realization by Tsirelson that quantum correlations are similarly limited~\cite{tsirelson_bound}: in fact, one can write over the paper physical theories compatible with special relativity, but capable of generating stronger-than-quantum correlations~\cite{Tsirelson85, PR_box}. Like their classical counterparts, limits on quantum correlations can be expressed in the form of quantum Bell inequalities, the statement that a certain linear function of the measurement statistics of a Bell experiment is bounded by some value if the experimental setup is subject to the laws of quantum mechanics.

Quantum Bell inequalities provide falsifiability criteria for quantum theory. In addition, an experimental Bell value close to the quantum limit can sometimes reveal important information regarding the physical system generating the correlations, such as the degree of entanglement of the underlying quantum state~\cite{ent_neg} or its local dimension~\cite{Brunner_2008}. In this regard, some quantum Bell inequalities are self-testing~\cite{self_testing_review}, meaning that their maximal value can only be achieved in quantum systems by measuring a specific multipartite quantum state with specific local measurement operators. 

There exist few general tools to construct quantum Bell inequalities~\cite{tsirelson_XOR,Wehner_2006,NPA1, NPA2}, and the computational complexity of them all is exponential on the number of parties. To tackle the problem of generating many-body quantum Bell inequalities, one of us (and collaborators) introduced the notion of \emph{connectors} \cite{connectors}. In the context of quantum nonlocality, connectors are directed tensors representing coarse-graining transformations within a physical theory whose states are vectors of quantum measurement probabilities. The acyclic contraction of a number of connectors thus defines a renormalization flow within the theory, and the resulting tensor constitutes a many-body quantum Bell inequality. Unfortunately, while connector theory leads to sound quantum Bell inequalities for arbitrarily many parties, those can sometimes be very loose~\cite{connectors}.

In this paper, we introduce the concept of \emph{tight connectors}, which have the property that any contraction thereof according to very simple rules defines a tight quantum Bell inequality. More specifically, each of the legs or indices of a tight connector is associated to a system of local measurements, and only contractions between legs with the same measurements (what we call congruent contractions) are allowed. Using existing quantum Bell inequalities, we construct large analytic families of tight connectors. For some such families, we show that their contraction results in quantum Bell inequalities of correlator form for which the quotient between their classical and quantum values decreases exponentially with the system size.

In addition, all the connectors we find are ``fully self-testing". This property implies that their congruent contraction defines self-testing quantum Bell inequalities~\cite{self_testing_review}. The self-tested measurement operators correspond to the measurements associated to the uncontracted legs and the self-tested states are tensor network states~\cite{tensor_network_states} with the same network topology as the connector tensor network. Our analytic families of tight connectors therefore allow us to derive analytic two input/two output quantum Bell inequalities that self-test arbitrary multi-qubit tree tensor network states of bond dimension 2, as well as arbitrary lists of pairs of local qubit projectors.

The structure of this paper is as follows: after some preliminary background, we introduce in Sec.~\ref{sec:tight_CI} the notion of tight connectors and show that their contraction leads to tight quantum Bell inequalities, maximally violated by tensor network states. For clarity, in Sec.~\ref{sec:firstcomplexes} we use the famous Tsirelson’s bound~\cite{tsirelson_bound} to construct an instance of a tight connector. Next, we argue that the same construction that led to this “Tsirelson connector” allows generating tight connectors from quantum Bell inequalities other than Tsirelson’s \cite{acin2012randomness,bamps2015sum,mckague2014self,baccari2020scalable,wooltorton2023device}, resulting in large analytic families of tight connectors. In Sec.~\ref{sec:XOR_expo} we prove that, for one of these families, the congruent contraction of connectors results in Bell functionals of the XOR type \cite{XOR_games} for which the quotient between the maximum quantum and classical values increases exponentially with the number of parties. We next define (Sec.~\ref{sec:self-testing}) fully self-testing quantum Bell inequalities, a strengthening of the familiar notion of self-testing. We provide an easily verifiable criterion that allows promoting self-testing to fully self-testing results, which we use to prove that all tight connectors previously identified are fully self-testing. In Sec.~\ref{sec:examples} we explore the properties of some of the new quantum Bell inequalities. Finally, we present our conclusions.

\section{Notation}
Some definitions follow. A $q$-partite measurement system $A$ is a structure of the form $\{A^{(k)}_{a|x}:a,x\}_{k=1}^q$, where $A^{(k)}_{a|x}\geq 0$ refers to the POVM element of party $k$, with measurement setting $x$ and outcome $a$. 
%Sometimes $A$ will represent operator variables; sometimes measurement operators, in which case we will add a bar, i.e., $\bar{A}$. 
A tensor polynomial of variables $\{A^{(k)}_{a|x}:a,x\}_{k=1}^q$ is a sum of products of expressions of the form
$
\id^{(1)}\otimes...\otimes \id^{(k-1)}\otimes A^{(k)}_{a|y}\otimes\id^{(k+1)}...\otimes\id^{(q)}\;.
$
The degree of the tensor polynomial will be given by the maximum number of products within all factors $k$. E.g.: the degree of the tensor polynomial $A^{(1)}_{0|1}\otimes A^{(2)}_{1|2}A^{(2)}_{2|3}+5\id^{(1)}\otimes A^{(2)}_{1|2}$ is $2$.

We will term quantum (classical) \emph{behaviours} those probability vectors $P(\vec{a}|\vec{x})$ that can be achieved in specific quantum (classical) realizations of a Bell experiment. For example, a bipartite no-signalling quantum behaviour is a vector of probabilities of the form
$P(a_1,a_2|x_1,x_2) \allowbreak  =\Tr{(A^{(1)}_{a_1|x_1}\otimes A^{(2)}_{a_2|x_2}) \rho^{(12)}}$, where $\rho^{(12)}$ is a bipartite quantum state. Each number $P(a_1,a_2|x_1,x_2)$ represents the probability that parties $1$ and $2$ respectively observe the results $a_1,a_2$ when they respectively conduct experiments $x_1,x_2$ in separate labs. In the following, we regard $q$-partite behaviors $P(a_1,...,a_q|x_1,...,x_q)$ as tensors with $q$ indices or legs. The set of possible values that each leg $k\in\{1,...,q\}$ can take corresponds to the set $\{(a_k|x_k)\}$ of possible output-input pairs of party $k$.

In the context of quantum Bell inequalities, a $q\to p$ connector ${\sf C}$~\cite{connectors} is a tensor of the form ${\sf C}_{b_1,...,b_p|y_1,...,y_p}^{a_1,...,a_q|x_1,...,x_q}$ which, contracted with the first $q$ legs of any $r$-partite quantum behavior tensor $P(\vec{a},\vec{c}|\vec{x},\vec{z})$ (where  $\vec{c},\vec{z}$ are $(r-q)$-dimensional vectors of outcomes and settings), generates a (possibly non-normalized) $p+r-q$ quantum behavior
\begin{equation}
P'(\vec{b},\vec{c}|\vec{y},\vec{z})=\sum_{\vec{a},\vec{x}}{\sf C}_{\vec{b}|\vec{y}}^{\vec{a}|\vec{x}}P(\vec{a},\vec{c}|\vec{x},\vec{z}).
\label{eq:conn_P_action}
\end{equation}
In this paper, we will only deal with $q\to 1$ connectors. Given one such connector ${\sf C}$ and a $q$-partite measurement system $A$, we denote by ${\sf C}[A]$ the set of operators (cf. Fig.~\ref{fig:intro_connector}):
\begin{equation}
{\sf C}[A]:=\bigg\{\sum_{\vec{a},\vec{x}}{\sf C}_{b|y}^{a_1,...,a_q|x_1,...,x_q}A^{(1)}_{a_1|x_1}\otimes...\otimes A^{(q)}_{a_q|x_q}:b,y\bigg\}.
\label{eq:conn_action}
\end{equation}
We will distinguish with brackets the action ${\sf C}[A]$ of a connector on a set of operators, from the dependence ${\sf C}(\CC)$ of a connector on a corresponding complex $\CC$ (cf. below).
The same notation will be used for Bell functionals: given a $q$-partite functional ${\sf B}^{\vec{a}|\vec{x}}$, the expression ${\sf B}[{A}]$ will denote the corresponding Bell operator.
In the following, we only consider normalized connectors, i.e., those connectors ${\sf C}$ such that $\sum_b{\sf C}_{b|y}[{A}]=\id,\forall y$, for all measurement systems ${A}$ (in the general case, $\sum_b{\sf C}_{b|y}[{A}]=\Lambda$, with $0\leq\Lambda\leq \id$ \cite{connectors}).
%\miguel{I believe that our results can be generalized to unnormalized connectors too, but we might leave this for another paper.}
This implies that, if $A$ is a $q$-partite system of POVMs, then ${\sf C}[A]$ is a (normalized) $1$-partite system of measurement operators.

For reference, these main choices of notation are summarized in Table~\ref{tab:conventions}.

\begin{table}[]
    \centering
    \begin{tabular}{c|c}
    Description & Notation \\ \hline
        (POVM) operator on party $j$ with output $a$ and input $x$ & $A^{(j)}_{a|x}$  \\
         Correlator basis for dichotomic measurements & $K^{(j)}_x:=A^{(j)}_{0|x}-A^{(j)}_{1|x}$\\
         Projective measurement & $M_{b|y}$\\
         Coisometry $V$ & $VV^\dagger=\id$\\
         Isometry $U$ & $U^\dagger U=\id$\\
        Connector ${\sf C}$ action on measurement operators & ${\sf {C}}[A]:=$Eq.\eqref{eq:conn_action} \\
         Bell functional ${\sf B}$ action on measurement operators, or correlators & ${\sf {B}}[A]$,  ${\sf {B}}[K]$\\
        Connector complex & $\CC:=\{{\sf C},\bar{A},V\}$\\
        Connector complex labelling & e.g. ${\sf C}(\CC)$, $V(\CC)$
    \end{tabular}
    \caption{Main notation conventions used in this paper.}
    \label{tab:conventions}
\end{table}

\begin{figure}[H]
    \centering
    \includegraphics[width=0.7\linewidth]{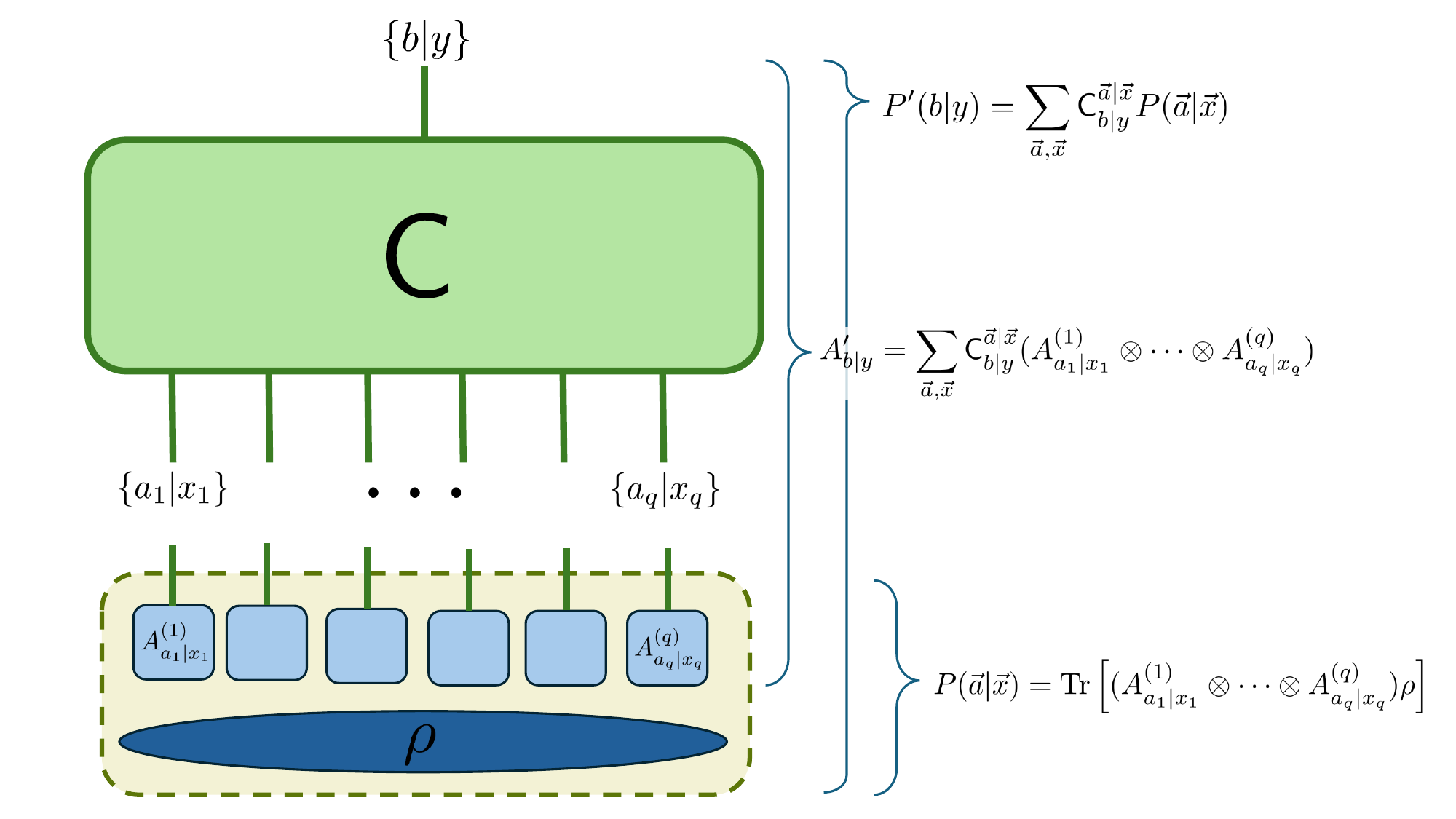}
    \caption{A connector ${\sf C}^{\vec{a}|\vec{x}}_{b|y}$ can be seen as a linear transformation between probabilities $P\rightarrow P'$ (via the contraction of the corresponding indices, cf. eq.~\eqref{eq:conn_P_action})  that preserves a set of behaviours, in our case, quantum behaviours allowed by the Born rule. The action of the connector can also be interpreted as a transformation on the measurement operators $A\rightarrow A'$ (cf. eq.~\eqref{eq:conn_action}), inducing a renormalization flow when repeated in a larger acyclic network.}
    \label{fig:intro_connector}
\end{figure}

\section{Tight connectors, tight inequalities}
\label{sec:tight_CI}
Consider an $q\to 1$ quantum connector ${\sf C}$ with the property that, for some Hilbert spaces $H_1,...,H_q$ and some measurement operators $\{\bar{A}^{(k)}_{a|x}:a,x\}\subset B(H_k)$, there exists a co-isometry $V:\bigotimes_k H_k\to H$ such that the operators $\bar{M}_{b|y}\in B(H)$, defined by
\begin{equation}
\bar{M}_{b|y}:=V{\sf C}_{b|y}[\bar{A}]V^\dagger,
\label{proj_meas}
\end{equation}
are projectors satisfying $\sum_b \bar{M}_{b|y}=\id_H$, for all $y$. Then we say that ${\sf C}$ is a \emph{tight connector} \footnote{Allowing $\bar{M}$ to be a general POVM would still allow the construction of tight quantum Bell inequalities, but would not allow self-testing in general.}. The name stems from the fact that, for all $\bar{M}_{b|y}\not=0,\id$, the quantum Bell inequalities
\begin{equation}
0\leq {\sf C}_{b|y}[A]\leq 1   
\label{tight_ineq}
\end{equation}
are tight. In particular, the value $1$ ($0$) can be achieved with the operators $\bar{A}$ and any state of the form $V^\dagger\ket{\phi}$, with $\bar{M}_{b|y}\ket{\phi}=\ket{\phi}$ ($\bar{M}_{b|y}\ket{\phi}=0$).

\begin{defin}[Connector complex]
We define a \emph{connector complex} as a tuple $\mathcal{C}:=({\sf C},\bar{A},V)$ (where ${\sf C}(\mathcal{C})$ is a connector; $\bar{A}$, a system of measurement operators; and $V$, a coisometry), such that eq.~\eqref{proj_meas} is satisfied for some projective measurement system $\bar{M}$ (implying that ${\sf C}$ is tight). In the following, we use ${\sf C}(\mathcal{C})$, $\bar{A}(\mathcal{C})$, $V(\mathcal{C})$ to denote the elements of the tuple $\mathcal{C}$. The symbols $\bar{M}_{b|y}(\mathcal{C})$, $H_k(\mathcal{C})$, $H(\CC)$ will similarly have obvious meanings. 
\end{defin}
Notice that, if $\CC=({\sf C},\bar{A},V)$ is a connector complex, then so is $\CC'=({\sf C},(U^{(k)}\bar{A}^{(k)} U^{(k)\dagger})_k,WV\bigotimes_k U^{(k)\dagger})$ for any unitaries $U^{(1)},...,U^{(q)}, W$. Complexes $\CC$ and $\CC'$ will be said unitarily equivalent.

Call $I(\CC)$ the space spanned by vectors of the form $V(\CC)^\dagger \ket{\beta}\in \bigotimes_k H_k$ (where $\ket{\beta}\in H$). Note that, since ${\sf C}_{b|y}[\bar{A}]$ is an effect and its submatrix $\bar{M}_{b|y}$ is a projector, it follows that ${\sf C}_{b|y}(\bar{A})I(\CC)\subset I(\CC)$, see Appendix \ref{app:invariant} for a proof. Thus we will call $I(\CC)$ the \emph{invariant subspace} of $\CC$.

Consider two connector complexes $\mathcal{C}_1$, $\mathcal{C}_2$, and suppose that we contract the output leg of ${\sf C}(\mathcal{C}_1)$ with the $k^{th}$ input legs of ${\sf C}(\mathcal{C}_2)$, as in Fig.~\ref{fig:basic_contraction}. We will call this contraction \emph{congruent} if $\bar{M}(\mathcal{C}_1)=\bar{A}^{(k)}(\mathcal{C}_2)$ (this means in particular that $\bar{A}^{(k)}(\mathcal{C}_2)$ is a system of projective measurements). Any congruent contraction of the connectors of two connector complexes results in a tight connector. 

\begin{figure}[H]
    \centering
    \includegraphics[width=0.75\linewidth]{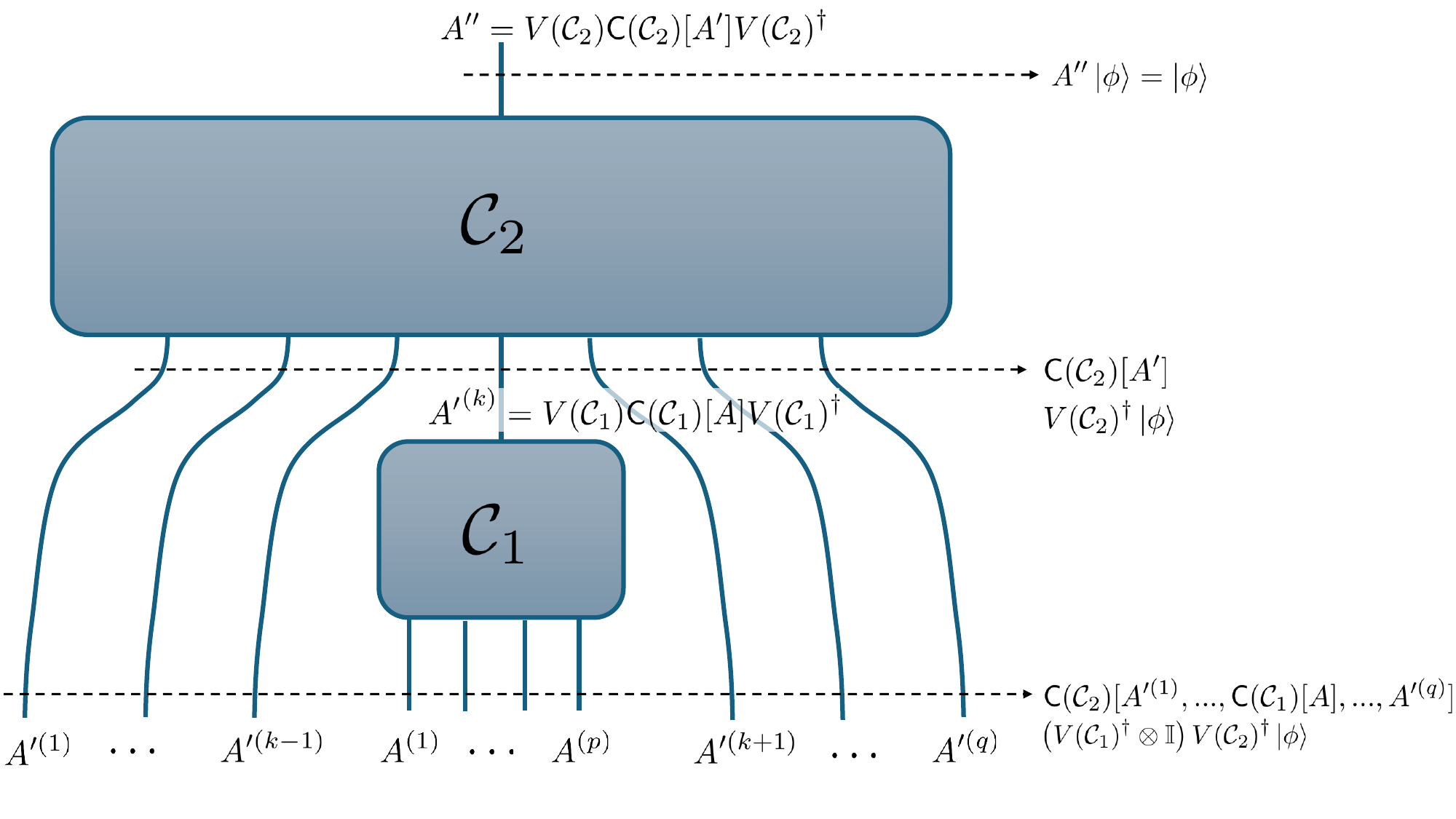}
    \caption{The result of using the output of a $p\rightarrow 1$ connector ${\sf C}(\CC_1)$ as the $k$-th input of a $q\rightarrow 1$ connector ${\sf C}(\CC_2)$ is a $p+q-1\rightarrow 1$ connector. A connector is tight if for some choice of measurements $\bar{A}$ and some coisometry $V$, the ``renormalized" measurement $A'=V{\sf C}[\bar{A}]V^\dagger=\bar{M}$ is projective. This ensures that the Bell inequality $0\leq {\sf C}[A]\leq 1$ is tight and can be saturated by applying $V^\dagger$ to the corresponding eigenvector of $M$ (see text). The set $\CC=\{{\sf C}, \bar{A}, V\}$ then defines a \emph{connector complex}. When the renormalized measurement coincides with the optimal input measurement $\bar{A}'$ of the following connector ${\sf C}(\CC_2)$, the complexes $\CC_1$ and $\CC_2$ can be \emph{congruently contracted}, resulting in another tight connector complex. This allows to back-propagate the saturation of any of the (trivial) quantum Bell inequalities associated to the final output leg -- say $A''\leq 1$ -- and its optimal eigenvector $\ket{\phi}$, to the corresponding quantum Bell inequalities on $A'$ and $A$ and their saturation strategies, via the sequential application of the isometries $V^\dagger_2$ and $V^\dagger_1$.}
    \label{fig:basic_contraction}
\end{figure}

Let us see why: call ${\sf C}$ the result of contracting the tensors ${\sf C}(\mathcal{C}_1), {\sf C}(\mathcal{C}_2)$ of the two complexes. Then we have that ${\sf C}$ is a connector, as it is the contraction of two connectors \cite{connectors}. Next, define the co-isometry
\begin{equation}
V=V(\CC_2)\left(\bigotimes_{j<k}\id^{(j)}\otimes V(\CC_1)\otimes \bigotimes_{j>k}\id^{(j)}\right),
\label{new_coiso}
\end{equation}
where $\id^{(j)}$ denotes the identity operator in the Hilbert space $(H_j(\CC_2))$. It is easy to see that $V$ satisfies eq. (\ref{proj_meas}), with $\bar{M}_{b|y}=\bar{M}_{b|y}(\CC_2)$ and 
\begin{align}
\bar{A}^{(j)}_{a|x}=&\bar{A}^{(j)}_{a|x}(\CC_2),\mbox{ for }j<k,\nonumber\\
&\bar{A}^{(j-k+1)}_{a|x}(\CC_1),\mbox{ for }k\leq j< k+q(\CC_1),\nonumber\\
&\bar{A}^{(j-q(\CC_1)+1)}_{a|x}(\CC_2),\mbox{ for } j\geq k+q(\CC_1).
\end{align}

Since $M(\CC_2)$ represent projective measurements, we have that the tuple $\CC=({\sf C},\bar{A}, V)$ is a connector complex. We will dub $\CC$ the (congruent) contraction of complexes $\CC_1$, $\CC_2$.

Now, think of a network of connector complexes, congruently contracted in a tree-like graph with just one root. By the previous argument, this defines a connector complex $\CC$, and so ${\sf C}={\sf C}(\CC)$ is tight. Hence, for any quantum Bell inequality of the form (\ref{tight_ineq}), the upper (lower) bound is guaranteed to be tight if $\bar{M}_{b|y}(\CC)\not=0$ ($\bar{M}_{b|y}(\CC)\not=\id$). Both upper and lower bounds can be achieved with measurements $\bar{A}(\CC)$ and states of the form $V^\dagger(\CC)\ket{\phi}$, with $\ket{\phi}$ being an eigenvector of $\bar{M}_{b|y}(\CC)$ with eigenvalue $1$ or $0$, depending on the bound we wish to saturate.

Remarkably, from the composition of co-isometries (\ref{new_coiso}), it follows that the optimal state $\ket{\psi}$ is a tensor network state: to generate it, one just needs to apply to $\ket{\phi}$ a sequence of local isometries, whose contraction structure mimics that of the connectors', see Figure \ref{fig:conn2tree} (for details refer to Appendix \ref{app: Tensor network state}).

\begin{figure}[H]
\centering
\includegraphics[width=0.6\textwidth]{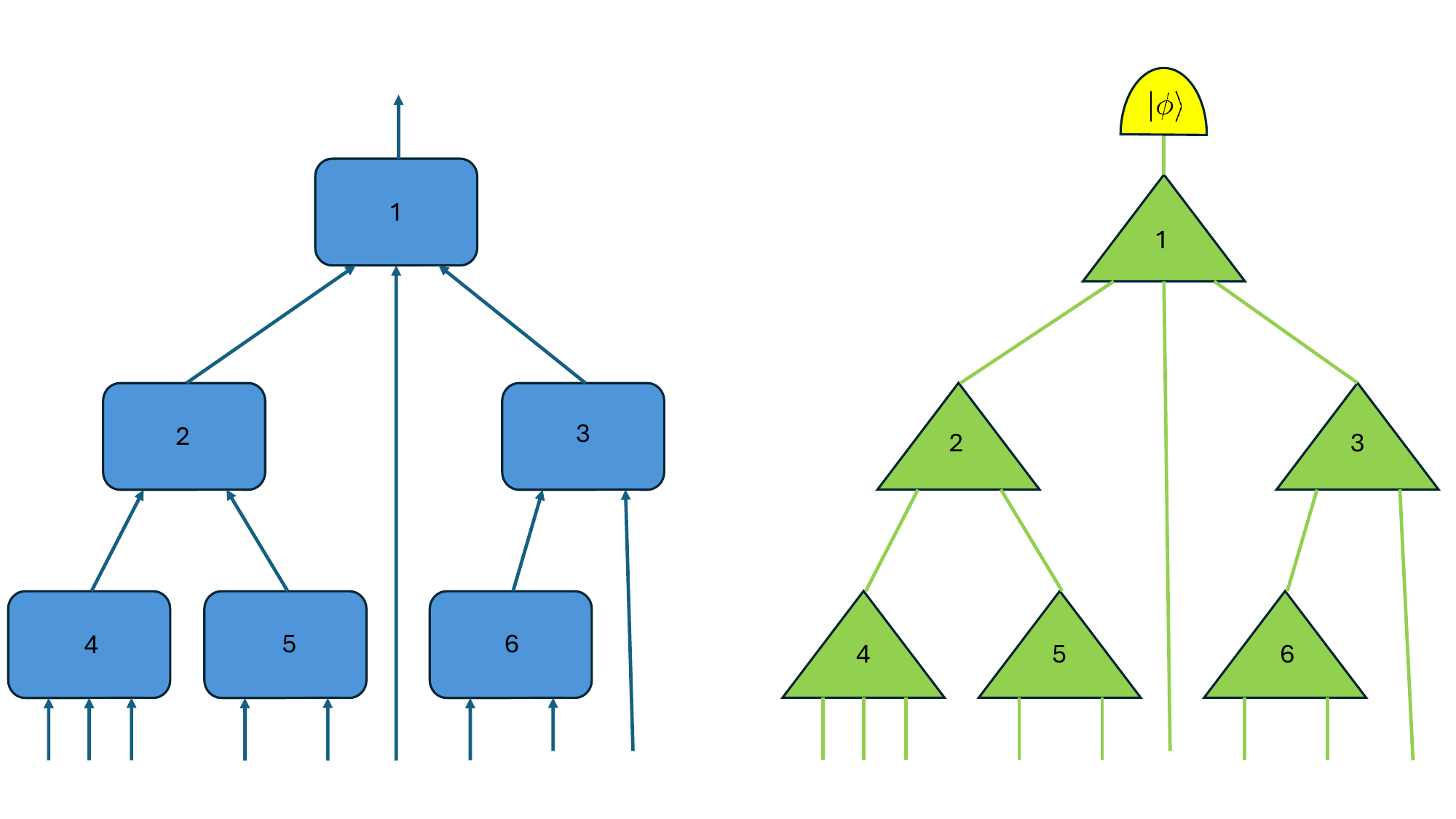}
\caption{Left: Bell functional obtained by congruently contracting a network of connector complexes. Right: Tensor network state maximizing the quantum Bell functional. Green triangles denote the adjoint of the co-isometry of the corresponding connector complex.}
\label{fig:conn2tree}
\end{figure}

\section{Examples of connector complexes}
\label{sec:firstcomplexes}
We now show the existence of a simple $2\rightarrow 1$ connector complex, based on the Clauser-Horne-Shimony-Holt (CHSH) inequality \cite{CHSH} and the corresponding Tsirelson bound~\cite{tsirelson_bound}, which self-tests maximally entangled 2-qubits Bell states and pairs of dichotomic operators represented by orthogonal vectors in the Bloch sphere \cite{Summers1987} (we later provide other examples involving partially entangled states, non-orthogonal measurements, as well as $q\rightarrow 1$ connector complexes with $q>2$).
We will work in Bell scenarios with only two outputs, labeled $+1$ and $-1$. Following tradition in the nonlocality literature, instead of representing the $x^{th}$ measurement of party $k$ with the POVM $\{A^{(k)}_{a|x}:a=\pm1\}$, we will use the POVM dichotomic operator $K^{(k)}_x:=A^{(k)}_{+1|x}-A^{(k)}_{-1|x}$.

\subsection{CHSH - Tsirelson complex $\CC_{\rm Tsi}$.}

Consider then the following two CHSH Bell functionals defined by
\begin{align}
\label{eq:CHSH_tsi}
    2\sqrt{2} \sB_0[K] &= -K^{(1)}_0\otimes K^{(2)}_0 + K^{(1)}_1 \otimes K^{(2)}_0 + K^{(1)}_0 \otimes K^{(2)}_1 + K^{(1)}_1 \otimes K^{(2)}_1\;,\\
    2\sqrt{2} \sB_1[K] &=\;\;\; K^{(1)}_0\otimes K^{(2)}_0 + K^{(1)}_1 \otimes K^{(2)}_0 + K^{(1)}_0 \otimes K^{(2)}_1 - K^{(1)}_1 \otimes K^{(2)}_1\;.
\end{align}
Tsirelson's bound~\cite{tsirelson_bound} asserts that both  are normalized as $-1\leq\sB_{0,1}\leq 1$.
Let $\sx,\sy,\sz$ denote the three Pauli matrices. Up to isometries, we express the optimal CHSH Bell operators saturating the normalization as
%\begin{align}
%    \bar{K}^{(k)}_0=\sx \;, \bar{K}^{(k)}_1=\sz,\; k=1,2.
%\end{align}
\begin{align}
\nonumber
    \bar{K}^{(1)}_0 &=\sigma_3\;,\; &\bar{K}^{(1)}_1&=\sigma_1\;,\\
    \bar{K}^{(2)}_0 &=\frac{\sigma_3 + \sigma_1}{\sqrt{2}}\;,\; &\bar{K}^{(2)}_1&=\frac{\sigma_3 - \sigma_1}{\sqrt{2}}\;,
    \label{eq:optimal_meas_CHSH}
\end{align}
The maximum and minimum CHSH values for $\sB_1$ are then obtained (in this basis) by the maximally entangled states
%\begin{align}
%    \ket{\phi_+}&=\frac{\alpha^{-1}\ket{00}+\alpha\ket{01}+\alpha\ket{10}-\alpha^{-1}\ket{11}}{2\sqrt[4]{2}}\;, \\
%    \ket{\phi_-}&=\frac{\alpha\ket{00}-\alpha^{-1}\ket{01}-\alpha^{-1}\ket{10}-\alpha\ket{11}}{2\sqrt[4]{2}}\;,\\
%    \text{where } \alpha &=(\sqrt{2}-1)^{\frac{1}{2}}\;.
%\end{align}
\begin{align}
     \ket{\phi_+} =\frac{\ket{00}+\ket{11}}{\sqrt{2}}\;,\quad
    \ket{\phi_-} =\frac{\ket{01}-\ket{10}}{\sqrt{2}}\;.
\end{align}
That means
\begin{align}
\label{eq:c1_like_Z}
    \sB_1[\bar{K}]\ket{\phi_\pm}=\pm\ket{\phi_\pm}\;.
\end{align}
Remarkably (and crucially for what follows) the following property also holds
\begin{align}
\label{eq:c0_like_X}
    \sB_0[\bar{K}]\ket{\phi_\pm}=\ket{\phi_\mp}\;.
\end{align}
That is, Eqs.~\eqref{eq:c1_like_Z} and~\eqref{eq:c0_like_X} explicitly show that $\sB_0$ and $\sB_1$ \emph{effectively act} as the $\sx$ ($X$) and $\sz$ ($Z$) Pauli matrices on the subspace spanned by $\{\ket{\phi_+},\ket{\phi_-}\}$. This can then be used to define our first example of a connector complex, by using the co-isometry
\begin{align}
    V=\ket{0}\bra{\phi_+}+\ket{1}\bra{\phi_-}
\end{align}
and  the connectors
\begin{align}
    \sC_{\pm|0}[{A}]= \frac{\id\pm \sB_0[{K}]}{2}\;,\quad
    \sC_{\pm|1}[{A}]= \frac{\id\pm \sB_1[{K}]}{2}\;.
\label{eq:first_tight_connectors}
\end{align}
These are valid connectors as can be seen from the chosen normalization of the Bell functionals, which is such that $0\leq \sC \leq 1$ for any quantum behaviour. The same choice guarantees that the connectors are tight, as
\begin{align}
    \bar{M}_{\pm|0} =V\frac{\id\pm \sB_0[\bar{K}]}{2}V^{\dagger}=\frac{\id\pm \sx}{2}\;,\quad
    \bar{M}_{\pm|1} =V\frac{\id\pm \sB_1[\bar{K}]}{2}V^{\dagger}=\frac{\id\pm \sz}{2}\;,
    \label{eq:tsirelson_Mbar}
\end{align}
define the associated projectors. Equations~(\ref{eq:CHSH_tsi}-\ref{eq:tsirelson_Mbar}) thus define the CHSH-Tsirelson connector complex $\CC_{\rm Tsi}$.
Crucially, the renormalized ``output measurements" $\bar{M}(\CC_{\rm Tsi})$~\eqref{eq:tsirelson_Mbar} can be then congruently contracted with connectors complexes $\CC$ that are saturated by orthogonal measurements of the same form, i.e. satisfying $\bar{A}(\CC)=\bar{M}(\CC_{\rm Tsi})$. In particular, from~\eqref{eq:optimal_meas_CHSH} and the observation that the pair of operators $(\bar{K}^{(2)}_0, \bar{K}^{(2)}_1)$ is unitarily equivalent to $(\bar{K}^{(1)}_0, \bar{K}^{(1)}_1)$, the CHSH - Tsirelson complex can be congruently contracted with itself or some unitarily equivalent version, depending on the contracted input leg. Since unitarily equivalent complexes share the same connector, it follows that contractions of multiple copies of the connector ${\sf C}(\CC_{\mbox{Tsi}})$ can only generate tight quantum Bell inequalities.

%\subsection{Other connector complexes}
It is now worth noticing that the construction leading to the $2\rightarrow 1$ connectors~\eqref{eq:first_tight_connectors} has a simple interpretation and can be repeated for all instances in which there exist two Bell functionals $-1\leq \sB_{0,1} \leq 1$ and a measurement system $\bar{A}$ such that $\sB_{0,1}[\bar{A}]$ can be projected on a common two-dimensional subspace on which they act as projective dichotomic operators, i.e., rotations of Pauli matrices. That is,
\begin{prop}
\label{prop:pauli_construction}
    Consider a set of normalized Bell functionals satisfying $-1\leq \sB_i[A]\leq 1$, an optimal set of measurements $\bar{A}$ and a $2$-dimensional subspace $\mathbb{S}$ such that each $\sB_i$ can be saturated in $\mathbb{S}$ on both hands, i.e. $\mathbb{S}\equiv {\sf Span}(\ket{\phi_{+,i}},\ket{\phi_{-,i}})$ $\forall i$ and $\sB_i[\bar{A}]\ket{\phi_{\pm,i}}=\pm\ket{\phi_{\pm,i}}$.
    Then, the connector defined by ${\sf C}_{\pm|i}[A]=\frac{\id\pm \sB_i[A]}{2}$ is tight and defines a connector complex $\{{\sf C},\bar{A},V\}$, with $V$ being a projection onto $\mathbb{S}$.
\end{prop}
%The output leg of such a complex can then be ``fed" (congruently contracted) to the input leg of any other connector complex that is saturated by isometric choices of Pauli measurements.
In the following, we show examples of this construction beyond the $2\rightarrow 1$ CHSH-Tsirelson connector presented above. In particular, we present 2-inputs 2-outputs couples of quantum Bell functionals that are saturated by: \emph{i)} all partially entangled 2-qubits states (Tilted CHSH~\cite{acin2012randomness,bamps2015sum}), \emph{ii)}  non-orthogonal measurements (WBC Bell inequalities~\cite{wooltorton2023device}), as well as \emph{iii)}  $q$-partite graph-states (inspired by the BASTA Bell inequalities~\cite{baccari2020scalable}) leading to $q\rightarrow 1$ tight connectors, and thus valid complexes via Proposition~\ref{prop:pauli_construction}.
We provide in the following a short presentation of such connectors, see Appendix~\ref{sec:other_complexes} for details.

%\vspace{0.5cm}
%\paragraph{Tilted CHSH complex ($\CC_{\rm Tilt}$) for non-maximally entangled qubit states.}
\subsection{Tilted CHSH complex ($\CC_{\rm Tilt}$) for non-maximally entangled qubit states.}
Consider the Tilted CHSH~\cite{acin2012randomness,bamps2015sum} operator functionals
\begin{align}
\label{eq:tilted_functional}
    \beta(\theta) {\sf B}_{0,\theta}[K] &= \alpha(\theta) K^{(1)}_1\otimes \id^{(2)}  -K^{(1)}_0\otimes K^{(2)}_0 + K^{(1)}_1 \otimes K^{(2)}_0 + K^{(1)}_0 \otimes K^{(2)}_1 + K^{(1)}_1 \otimes K^{(2)}_{1}\;,\\
    \beta(\theta) {\sf B}_{1,\theta}[K] &= \alpha(\theta) K^{(1)}_0\otimes \id^{(2)} + K^{(1)}_0\otimes K^{(2)}_0 + K^{(1)}_1 \otimes K^{(2)}_0 + K^{(1)}_0 \otimes K^{(2)}_1 - K^{(1)}_1 \otimes K^{(2)}_1\;.
\end{align}
Here we follow the notation of~\cite{bamps2015sum}, parametrizing
$    \alpha(\theta)=2\left(1+2\tan(2\theta)^{2}\right)^{-\frac{1}{2}}
$.
The quantum bound $\beta(\theta)$  of the corresponding tilted CHSH inequality is given by $
    \beta(\theta)=\sqrt{8+2\alpha^{2}(\theta)}$, so 
that
$-1\leq \sB_{0,1}\leq 1\;$
for any choice of dichotomic operators $K$ (the classical bound is instead given by $\beta^{\rm cl}(\theta)=2+\alpha(\theta)$).
As it turns out, the $\pm 1$ values are (uniquely, cf.~\cite{bamps2015sum} and~\ref{secapp:tilted_CHSH}) saturated in $\sB_1$
when choosing the measurements
\begin{align}
\nonumber
    \bar{K}^{(1)}_0 &=\sigma_3\;,\; &\bar{K}^{(1)}_1&=\sigma_1\;,\\
    \bar{K}^{(2)}_0 &=\cos\mu(\theta)\, \sigma_3 + \sin\mu(\theta)\, \sigma_1\;,\; &\bar{K}^{(2)}_1&=\cos\mu(\theta)\, \sigma_3 - \sin\mu(\theta)\, \sigma_1\;,
\end{align}
with $\mu(\theta)=\arctan{(\sin{(2\theta)})}$ and
acting on the span of the non-maximally entangled qubit states
\begin{align}
    \ket{\phi_+} =\cos\theta\ket{00}+\sin\theta\ket{11}\;,\quad
    \ket{\phi_-} =\sin\theta\ket{01}-\cos\theta\ket{10}\;.
\end{align}
With this choice, one again has
\begin{align}
    \sB_1[\bar{K}]\ket{\phi_\pm}=\pm\ket{\phi_\pm}\;,\quad
    \sB_0[\bar{K}]\ket{\phi_\pm}=\ket{\phi_\mp}\;,
\end{align}
from which tight connector complexes $\CC_{\rm Tilt}$ can be built via~\eqref{eq:first_tight_connectors}.

%\vspace{0.5cm}
%\paragraph{WBC complex ($\CC_{\rm WBC}$) for non-orthogonal measurements.} 
\subsection{WBC complex ($\CC_{\rm WBC}$) for non-orthogonal measurements.}
In Ref.~\cite{wooltorton2023device}, the following family of Bell functionals is defined: 
\begin{align}
  {\beta(\theta,\varphi,\omega)}\sB_{1}(\theta,\varphi,\omega)[K] := &\cos(\theta + \varphi)\cos(\theta + \omega)\big(\cos(\omega) K^{(1)}_0 \otimes K^{(2)}_0  - \cos(\varphi) K^{(1)}_0 \otimes K^{(2)}_{1}\big) + \nonumber\\ 
   &\cos(\varphi)\cos(\omega)\big(\cos(\theta + \varphi) K^{(1)}_{1}\otimes K^{(2)}_{1} - \cos(\theta + \omega) K^{(1)}_{1}\otimes K^{(2)}_{0} \big) ,
\end{align}
with $\beta(\theta,\varphi,\omega):=\sin(\theta)\sin(\omega - \varphi)\sin(\theta + \varphi + \omega)$.
as shown in~\cite{wooltorton2023device}, if the angles $\theta, \varphi, \omega \in \mathbb{R}$ satisfy the condition $\cos(\theta + \varphi)\cos(\varphi)\cos(\theta + \omega)\cos(\omega) < 0$,
then $-1\leq\sB_{1}(\theta,\varphi,\omega)\leq 1$ is also normalized. Moreover, the extremal values of  $\sB_{1}$ are achieved for the state and operators
\begin{align}
\ket{\phi_{\pm}} &=\frac{1}{\sqrt{2}}(\ket{00}\pm i\ket{11})\;, & & \nonumber\\
\bar{K}^{(1)}_0 &=\sigma_1\;, & &\bar{K}^{(1)}_1=\cos(\theta)\sigma_1+\sin(\theta)\sigma_2\;,\nonumber\\
\bar{K}^{(2)}_0 &=\cos(\varphi)\sigma_1+\sin(\varphi)\sigma_2\;,& &  \bar{K}^{(2)}_1=\cos(\omega)\sigma_1+\sin(\omega)\sigma_2\;.
\label{ref_W}
\end{align}
Consider now the complementary functional 
\begin{equation}
{\sB}_{0}(\theta,\varphi,\omega)[K^{(1)}_0,K^{(1)}_1,K^{(2)}_0,K^{(2)}_1]:=\sB_{1}(\theta,\varphi,\omega)[K^{(1)}_1,K^{(1)}_0,K^{(2)}_1,K^{(2)}_0].
\end{equation}
This is obviously another normalized quantum Bell inequality. One can verify that not only $\sB_{0}(\theta,\varphi,\omega)[\bar{K}]\ket{\phi_\pm}=\pm \ket{\phi_\pm}$, but also
\begin{align}
    \sB_{0}(\theta,\varphi,\omega)[\bar{K}]\ket{\phi_+}& =-\cos(\theta+\varphi+\omega)\ket{\phi_+}+i\sin(\theta+\varphi+\omega)\ket{\phi_-}\;\\
    \sB_{0}(\theta,\varphi,\omega)[\bar{K}]\ket{\phi_-}& =\cos
    (\theta+\varphi+\omega)\ket{\phi_-}-i\sin(\theta+\varphi+\omega)\ket{\phi_+}\;.
\end{align}
Differently from the previous two examples, we can see that in this case $\sB_1$ effectively acts as $\sigma_3$ on the span of $\{\ket{\phi_-},\ket{\phi_+}\}$, while $\sB_0$ acts as $-\cos(\theta+\varphi+\omega)\sigma_3-\sin(\theta+\varphi+\omega)\sigma_2$. This implies that the connector defined by $\frac{\id\pm \sB_{1}(\theta,\varphi,\omega)[{K}]}{2}$ and $\frac{\id\pm \sB_{0}(\theta,\varphi,\omega)[{K}]}{2}$ is again tight, and can be congruently contracted with connector complexes that are saturated by couples of Pauli measurements separated by a relative angle $\theta+\varphi+\omega+\pi$ in the Bloch sphere.

%\vspace{0.5cm}
%\paragraph{BASTA complex ($\CC_{\rm BASTA}$) for q-partite graph states.}
\subsection{BASTA complex ($\CC_{\rm BASTA}$) for q-partite graph states.}
\label{subsec:C_basta}
As a final example, for the construction of $q\rightarrow 1$ tight connectors, we consider here the BASTA-Bell functional introduced in~\cite{baccari2020scalable} for graph-states.
Consider in particular a given graph $\G:=\{\V,\E\}$ defined by a set of $q$ vertices $\V$ and edges $\E$ and its corresponding graph state $\ket{\phi_+}$. By selecting a given vertex, say $v_1$, with $n_1$ graph-neighbours, an associated Bell functional is introduced as
\begin{align}
%\nonumber
   \beta_1(q,n_1) \sB_1(\G)[K] &= n_1 (K_0^{(1)}+K_1^{(1)})\bigotimes_{i\in n(1)} K_1^{(i)} + \sum_{i\in n(1)}  (K_0^{(1)}-K_1^{(1)})\otimes K_0^{(i)}\bigotimes_{1\neq j\in n(i)} K_1^{(j)} %+\\&
    + \sum_{j \notin n(1)\cup \{1\}} K_0^{(j)}\bigotimes_{k\in n(j)}K^{(k)}_1 \;,
    \label{eq:BASTA}
\end{align}
where $n(i)$ identifies neighbouring vertices of $v_i$, and $\beta(q,n_1)=(2\sqrt{2}-1)n_1+q-1$.
This Bell inequality is quantum-normalized as $-1\leq\sB_1(\G)[K]\leq 1$.
The $+1$ value is obtained (uniquely) by the corresponding graph-state $\ket{\phi_+}$, and measurements $\bar{K}$ chosen to in order to make each term in~\eqref{eq:BASTA} proportional to the stabilizer operators of $\ket{\phi_+}$
(see details in \cite{baccari2020scalable} or App.~\ref{app:graph}), which define a basis $G_i$ of operators satisfying $G_i\ket{\phi_+}=\ket{\phi_+}$. One can then verify that a complementary \emph{anti}-graph state $\ket{\phi_-}$ state exists satisfying, $G_i\ket{\phi_-}=-\ket{\phi_-}$. It follows that with the same choice of measurements the opposite hand $B_1(\G)[\bar{K}]\ket{\phi_-}=-\ket{\phi_-}$ is saturated. 
Finally, we construct a nontrivial complementary (normalized) Bell functional as
\begin{align}
    \beta_0(n_1)\sB_0(\G)[K] &= n_1 (K_0^{(1)}-K_1^{(1)})\bigotimes_{i>1} K_1^{(i)} - \sum_{i\in n(1)}  (K_0^{(1)}+K_1^{(1)})\otimes K_0^{(i)}\bigotimes_{j\in {\sf EN}_i} K_1^{(j)}
     \label{eq:BASTA'}
\end{align}
with $\beta_0(n_1)=2\sqrt{2}n_1$, ${\sf EN}_i$ being the set of vertices sharing an edge with both $v_1$ and $v_i$ or neither, and we show that is satisfies $\sB_0(\G)[\bar{K}] \ket{\phi_{\pm}}=\ket{\phi_\mp}$ (see App.~\ref{app:graph}).

\section{XOR Bell functionals and exponential quantum-to-classical ratio}
\label{sec:XOR_expo}
In this section, we prove that, if we contract a network of a special type of connector complexes, which includes the $\CC_{\rm WBC}$ family, then we can very easily bound the classical value of the associated quantum Bell inequalities. In turn, this will allow us to prove that the ratio between the maximum quantum and classical values of the Bell functional increases exponentially with the number of parties.

Throughout this section, we will just consider connectors and Bell functionals with indices $a|x$ such that $a$ can just take the values $-1,1$. In this case, it will be more convenient to express the outgoing index of each connector in a basis that, rather than referring to the measurement operators $\{A_{a|x}\}$, points to the corresponding dichotomic operators $\{K_\star=\id\}\cup\{K_x\}$, where $K_x:=A_{+1|x}-A_{-1|x}$. This choice of basis, which we implicitly already used in Sec.~\ref{sec:firstcomplexes}, formally consists in relabeling the values of the outgoing index or leg of a connector from $a|x$, to $\star,x$, with $\star$ ($x$) representing $\id$ ($K_x$). The change of basis for a connector ${\sf C}$ is as follows:
\begin{equation}
{\sf C}_{\star}:=\sum_{b=-1,1}{\sf C}_{b|y},{\sf C}_{y}:=\sum_{b=-1,1}{\sf C}_{b|y}b.
\label{proj2corr}
\end{equation}
The inverse transformation is given by:
\begin{equation}
{\sf C}_{b|y}=\frac{1}{2}({\sf C}_{\star}+b{\sf C}_{y}).
\label{corr2proj}
\end{equation}

Now, a $q$-partite Bell functional containing only products of $K_x$ terms, i.e. of the form
\begin{equation}
\label{eq:XOR_class}
\sum_{x_1,...,x_q}c_{x_1,...,x_q}\langle \bigotimes_{k=1}^q K^{(k)}_{x_k}\rangle\;,
\end{equation}
is called an \emph{XOR} or \emph{full correlation} Bell functional \cite{scarani2019bell}. Analogously, we call the tensor ${\sf C}$ an \emph{XOR connector} if ${\sf C}$ is a connector with dichotomic outcomes on each leg and 
$
{\sf C}_{\star}[K]=1,{\sf C}_{y}[K]={\sf B}_y[K],
%\label{def_XOR_connector}
$
where $\{{\sf B}_y\}_y$ are XOR Bell functionals. Similarly, any complex containing an XOR connector will be dubbed an \emph{XOR complex}. Note that all elements of the family $\CC_{\rm WBC}$ are XOR connector complexes. It is easy to verify that the congruent contraction of any pair of XOR connector complexes is another XOR complex.

The next two propositions will be very useful.

\begin{prop}
\label{prop:XOR_ineq}
For any XOR complex $\CC$, the inequality
\begin{equation}
-1\leq \left\langle {\sf C}_y(\CC)[K]\right\rangle_P\leq 1
\label{ineq_corr}
\end{equation}
holds for all quantum behaviors $P$. Moreover, the inequality is tight.

\end{prop}
\begin{proof}
The inequality follows from the fact that ${\sf C}_y(\CC)={\sf C}_{1|y}(\CC)-{\sf C}_{-1|y}(\CC)$, i.e., it is the difference between two quantum Bell inequalities, each of which can only take values in $[0,1]$. The tightness of (\ref{ineq_corr}) follows from a simple observation: by hypothesis $V(\CC){\sf C}_y(\CC)[\bar{A}(\CC)]V(\CC)^\dagger=\bar{M}_{1|y}(\CC)-\bar{M}_{-1|y}(\CC)$. Since $\sum_b\bar{M}_{b|y}(\CC)=1$, at least one of the two projectors is non-zero; suppose it is $\bar{M}_{1|y}(\CC)$. This implies immediately (cf. Sec.~\ref{sec:tight_CI}) that the the upper bound in ~\eqref{ineq_corr} can be saturated. However, for any achievable value $c$ of XOR game~\eqref{eq:XOR_class}, $-c$ is also achievable by flipping the sign of any of the outputs $K^{(k)}$. Therefore both ends of~\eqref{ineq_corr} are tight.

%The tightness of (\ref{ineq_corr}) follows from this observation: $V(\CC){\sf C}_y(\CC)[\bar{A}(\CC)]V(\CC)^\dagger=\bar{M}_{1|y}(\CC)-\bar{M}_{-1|y}(\CC)$. Since $\sum_b\bar{M}_{b|y}(\CC)=1$, at least one of the two projectors is non-zero; suppose it is $\bar{M}_{1|y}(\CC)$. Next, choose a state $\ket{\bar{\phi}}$ such that $\bar{M}_{1|y}(\CC)\ket{\bar{\phi}}=\ket{\bar{\phi}}$. Then we have that the state $\ket{\bar{\psi}}:=V(\CC)^\dagger\ket{\bar{\phi}}$ and the operators $\bar{A}:=\bar{A}(\CC)$ saturate the upper bound of the inequality. Saturating the lower bound is easy: take the same state $\ket{\bar{\psi}}$ and the measurement operators $\tilde{A}^{(k)}_{a|x}:=\bar{A}^{(k)}_{a|x}$, for $k\not=1$, $\tilde{A}^{(1)}_{a|x}:=\bar{A}^{(1)}_{-a|x}$. It follows that $\tilde{K}^{(1)}_x= -\bar{K}^{(1)}_x$, $\tilde{K}^{(k)}_x= \bar{K}^{(k)}_x$, for $k\not=1$. Thus, 
%\begin{equation}
%\bra{\bar{\psi}}{\sf C}_x(\CC)[\tilde{K}]\ket{\bar{\psi}}= -\bra{\bar{\psi}}{\sf C}_x(\CC)[\bar{K}]\ket{\bar{\psi}}=-1,
%\end{equation}
%where the first equality follows from the fact that XOR functionals are multilinear.
\end{proof}

\begin{lemma}
\label{lem:indep_ident}
For some XOR complex $\CC$, consider the Bell functional ${\sf C}_y(\CC)$, for some $y$. If $\CC$ is the result of congruently contracting several normalized XOR complexes $\CC^1,...,\CC^n$, then, for any $\vec{\lambda}\in\R^n$, we would have obtained the same Bell functional had we contracted the tensors $\tilde{{\sf C}}^1,...,\tilde{{\sf C}}^n$, with
\begin{equation}
\tilde{{\sf C}}^j_{\star}[K]:=\lambda^j, \tilde{{\sf C}}^j_{y}[K]:={\sf C}_y(\CC^j)[K],
\end{equation}
instead. 
\end{lemma}
\begin{proof}
Consider two tensors $\tilde{{\sf C}}^1$, $\tilde{{\sf C}}^2$ such that $\tilde{{\sf C}}^1_y[K],\tilde{{\sf C}}^2_z[K]$ are XOR Bell functionals, for all $y,z$, and $\tilde{{\sf C}}^1_{\star}[K]=\lambda\in\R$. Call $\tilde{{\sf C}}^3$ the result of contracting the $k^{th}$ incoming leg of $\tilde{{\sf C}}^2$ with the outgoing leg of $\tilde{{\sf C}}^1$ and assume that $\tilde{{\sf C}}^2$ has $q$ input legs. Then, we have that $\tilde{{\sf C}}^3_{y}[K]$ is independent of $\lambda$ for all $y$: otherwise, the functional $\tilde{{\sf C}}^2_{y}[K]$ would contain at least one term of the form 
\begin{equation}
\left\langle K^{(1)}_{x_1}\otimes... K^{(k-1)}_{x_{k-1}}\otimes\id_k\otimes K^{(k+1)}_{x_{k+1}}\otimes...K^{(q)}_{x_q}\right\rangle,
\end{equation}
i.e., it would not be an XOR game. By induction, the result of contracting a tree network of such modified XOR connectors will result in a modified XOR connector, whose associated dichotomic quantum Bell inequalities will be independent of the values of the constituents tensors at $\star$.    
\end{proof}

Now, consider a network of congruently contracted XOR connector complexes $\CC^1,...,\CC^{n_c}$, as in Fig.~\ref{fig:conn2tree}, and assume all components to be $q\rightarrow 1$ connectors with $q\leq \bar{q}$ (in the Figure $\bar{q}=3$). In such a case, the number $n_c$ of connectors needed in order to have a single final output leg, starting from $N$ initial inputs, has to satisfy 
\begin{align}
    n_c\geq \left\lceil \frac{N-1}{\bar{q}-1} \right\rceil\;.
    \label{eq:nc_linear_N}
\end{align} 
In particular, $n_c$ is of order $\mathcal{O}(N)$. 

Call $\CC$ the final XOR complex. We wish to estimate the local or classical value of ${\sf C}_y(\CC)$. Given a XOR connector ${\sf C}$, denote by $\gamma({\sf C})$ the maximum number such that
\begin{equation}
-1\leq \gamma({\sf C})\langle{\sf C}_{y}[K]\rangle_P\leq 1
\end{equation}
for all $y$ and all classical behaviors $P$. Of course, $\gamma({\sf C})\geq 1$ and it simply correponds to the ratio between the quantum and classical maximum values of the Bell functional $|{\sf C}_y|$. This quantity can be computed via, e.g., direct evaluation of $\{\langle{\sf C}_{x}[K]\rangle\}_y$ on all deterministic vertices of the classical set of correlations. For instance, one clearly has $\gamma(\sC(\CC_{\rm Tsi}))=\sqrt{2}$, while from \cite{wooltorton2023device}, we have that
\begin{align}
&\gamma({\sf C}(\CC_{\rm WBC}(\theta,\phi,\omega))\nonumber\\
&=\frac{1}{\beta(\theta,\phi,\omega)}\max_{\pm}\left|\cos(\theta+\omega)\cos(\omega)(\cos(\theta+\phi)\pm\cos(\phi))\right|+\left|\cos(\theta+\phi)\cos(\phi)(\cos(\theta+\omega)\pm\cos(\omega))\right|.
\end{align}

For any XOR connector ${\sf C}$, we define the classically-renormalized tensor $\mbox{Cl}({\sf C})$ as:
\begin{equation}
\mbox{Cl}({\sf C})_{\star}[K]=1,\mbox{Cl}({\sf C})_{y}[K]=\gamma({\sf C}){\sf C}_{y}[K].
\end{equation}
Since $\{\mbox{Cl}({\sf C})_{b|y}\}_{b,y}$ are normalized (standard) Bell inequalities, with 
$\sum_b\mbox{Cl}({\sf C})_{b|y}[K]=\mbox{Cl}({\sf C})_{\star}[K]=1$,    
it follows that $\mbox{Cl}({\sf C})$ is a normalized \emph{classical connector} \cite{connectors}, i.e., a tensor that, tensored with the identity map, maps normalized classical distributions to normalized classical distributions. Its contraction with other classical connectors thus generates standard Bell inequalities \cite{connectors}.

Now, call ${\sf C}^1,...,{\sf C}^{n_c}$ the respective connectors of the constituting complexes $\CC^1,...,\CC^{n_c}$. Then we have that ${\sf C}:={\sf C}(\CC)=f({\sf C}^1,...,{\sf C}^{n_c})$ for some multilinear function $f$ dependent on the exact contraction of the complexes. From all the above, it holds that
\begin{align}
{\sf C}_{y}&=f_y({\sf C}^1,...,{\sf C}^{n_c})\overset{\mbox{(Lemma \ref{lem:indep_ident})}}{=}f_y\left(\frac{1}{\gamma({\sf C}^1)}\mbox{Cl}({\sf C}^1),...,\frac{1}{\gamma({\sf C}^{n_c})}\mbox{Cl}({\sf C}^{n_c})\right)\nonumber\\
&\overset{\mbox{(multilinearity)}}{=}\frac{1}{\prod_{k=1}^{n_c}\gamma({\sf C}^k)}f_y\left(\mbox{Cl}({\sf C}^1),...,\mbox{Cl}({\sf C}^{n_c})\right).
\end{align}
The tensor $f_y\left(\mbox{Cl}({\sf C}^1),...,\mbox{Cl}({\sf C}^{n_c})\right)$ is the result of contracting the classical connectors in the argument and therefore it is a normalized Bell inequality. 

The identity above and Proposition \ref{prop:XOR_ineq} implies the main result of this section.
\begin{prop}
Let the complex $\CC$ be the result of congruently contracting the XOR complexes $\CC^1,...,\CC^{n_c}$. Then it holds that 
\begin{equation}
-1\leq \left\langle{\sf C}_y(\CC)[K]\right\rangle_P\leq 1,
\end{equation}
for all quantum behaviors $P$, and both inequalities can be saturated. In addition,
\begin{equation}
-\frac{1}{\Gamma}\leq \left\langle{\sf C}_{y}[K]\right\rangle_P\leq \frac{1}{\Gamma},
\label{classical_bound}
\end{equation}
for all classical behaviors $P$, with
\begin{equation}
\Gamma=\prod_{k=1}^{n_c}\gamma({\sf C}(\CC^k)).
\end{equation}
\end{prop}

As a corollary, if we congruently contract a network of XOR complexes $\{\CC^j\}_{j=1}^{n_c}$ (e.g., members of $\CC_{\rm WBC}$), with $\gamma({\sf C}(\CC^j))\geq \gamma_0>1$, for all $j$, then the maximum classical value of any of the associated normalized XOR quantum Bell inequalities will be upper bounded by $\gamma_0^{-n_c}=\gamma_0^{-\mathcal{O}(N)}$. The quotient between the maximum quantum and classical values will thus grow at least exponentially with $N$.

The reader might wonder under which circumstances one can assert that the bound (\ref{classical_bound}) is tight. One way to approach this problem is by defining an analog of tight connectors, namely, \emph{classical tight connectors}. Like their quantum counterparts, any such connector ${\sf C}$ would have a $1$-partite measurement system $\hat{A}^{(k)}$ associated to each input leg $k$. Crucially, those would be $1$-dimensional, meaning that, $\hat{A}^{(k)}_{a|x}=p(a|x,k)$, for all $x, k$, for some conditional probability distribution $p$. If $\hat{M}_{b|y}:={\sf C}_{b|y}(\hat{A})$ is a deterministic distribution, i.e., if $\hat{M}_{b|y}=\delta_{b,\beta_y}$, for some vector of outcomes $(\beta_y)_y$, then the connector will be called classically tight. If the XOR classical connectors $\mbox{Cl}({\sf C}^1),...,\mbox{Cl}({\sf C}^{n_c})$ are classically tight and congruently contracted, then the bound (\ref{classical_bound}) will be tight. An example of a classical tight connector is ${\sf C}^{\rm CHSH}:=\mbox{Cl}({\sf C}(\CC_{\rm Tsi}))$; this can be seen, e.g., by taking $\hat{A}_{1|x}=1, \hat{A}_{-1|x}=0$, for all $x$, in which case ${\sf C}^{\rm CHSH}_{1|y}(\hat{A})=1, {\sf C}^{\rm CHSH}_{-1|y}(\hat{A})=0$. Since input and output legs have the same classical outcomes, it follows that the bound~(\ref{classical_bound}) is tight for any Bell functional generated by contracting copies of $\CC_{\rm Tsi}$.

\section{Self-testing}
\label{sec:self-testing}
In this section, we show how the Bell inequalities constructed via the congruent contraction of connector complexes can be used to \emph{self-test} the underlying quantum state and measurements. Self-testing refers to the remarkable feature of certain Bell inequalities, whose maximum violation somehow uniquely fixes the experimental state and local operators, modulo local isometries, see~\cite{self_testing_review} for a review. A few results have been obtained in the literature, e.g. notably the fact that all pure bipartite entangled states can be self-tested~\cite{coladangelo2017all}, as well as all real projective measurements \cite{self_testing_all_meas}. More recently, the authors of~\cite{balanzójuandó2024allpuremultipartite} proved that all pure multipartite qubit states can be self-tested up to complex conjugation.

A formal definition of self-testing would be:
\begin{defin}[Standard (projective) self-testing]
\label{def:standard_self_testing}
Let $\bar{A}$ be a $q$-partite collection of projective measurements, with Hilbert spaces $H_1,...,H_q$, and let $\ket{\bar{\psi}}\in\bigotimes_k H_k$ be a normalized vector. Given a $q$-partite normalized Bell functional ${\sf B}$, we say that it \emph{fully self-tests} $(\ket{\bar{\psi}},\bar{A})$ if ${\sf B}[P]=1$ implies that, for any quantum realization $\{\psi,A\}$, with $A$ being a system of projective measurements, there exist local isometries $U_1,...,U_q$ such that
\begin{equation}
U_1\otimes...\otimes U_q\bigotimes_k A^{(k)}_{a_k|x_k}\ket{\psi}=\bigotimes_k \bar{A}^{(k)}_{a_k|x_k}\ket{\bar{\psi}}\ket{\mbox{junk}}
\label{st_corresp_standard}
\end{equation}
holds for all $a_1,x_1,...,a_q,x_q, i_1,...,i_q$.    
\end{defin}

For our purposes, we need to introduce a notion of self-testing that is slightly stronger.

\begin{defin}[Full Self-Testing]
\label{def:FST}
Let $\bar{A}$ be a $q$-partite collection of projective measurements, with Hilbert spaces $H_1,...,H_q$, and let $\ket{\bar{\psi}}\in\bigotimes_k H_k$ be a normalized vector such that
\begin{equation}
\{f[\bar{A}]\ket{\bar{\psi}}:f, \mbox{ tensor polynomial}\}=\bigotimes_k H_k.
\label{reconst_prop}
\end{equation}
Given a $q$-partite normalized Bell functional ${\sf B}$, we say that it \emph{fully self-tests} $(\ket{\bar{\psi}},\bar{A})$ if ${\sf B}[P]=1$ implies that, for any quantum realization $\{\psi,A\}$, with 
${A}^{(k)}_{a|x}:H_k'\to H_k'$, 
%for all $k$, of the $q$-partite distribution $P$, 
there exist an orthonormal basis $\{\ket{i_k}:i_k\}$ for each $H_k$, and orthonormal states 
$\{\ket{\psi(i_1,...,i_q)}:i_k=1,...,\mbox{dim}(H_k)\}\subset {H'}:=\bigotimes_k H'_k$ 
%$\{\ket{\psi(i_1,...,i_q)}:i_k=1,...,\mbox{dim}(\bar{H}_k)\}\subset \bar{H}:=\bigotimes_k \bar{H}_k$ 
satisfying
\begin{equation}
\ket{\psi}=\sum_{\vec{i}}\braket{i_1,...,i_q}{\bar{\psi}}\ket{\psi(i_1,...,i_q)},
\label{state_decomp}
\end{equation}
and local isometries $U_1,...,U_q$ such that
\begin{equation}
U_1\otimes...\otimes U_q\bigotimes_k A^{(k)}_{a_k|x_k}\ket{\psi(i_1,...,i_q)}=\bigotimes_k \bar{A}^{(k)}_{a_k|x_k}\ket{i_1,...,i_q}\ket{\mbox{junk}}
\label{st_corresp}
\end{equation}
holds for all $a_1,x_1,...,a_q,x_q, i_1,...,i_q$.
    
\end{defin}
\begin{remark}
Contrarily to standard self-testing (definition \ref{def:standard_self_testing}), in full self-testing the `actual' measurement system $A$ is not assumed to be projective; on the contrary, for any $k,x$, we allow $(A^{(k)}_{a|x})_{a}$ to be a general POVM. 
    
\end{remark}

\begin{remark}
\label{remark:st_tensor_pol}
Eqs. (\ref{state_decomp}), (\ref{st_corresp}) imply that, for any tensor polynomial $f[A]$, it holds that
\begin{equation}
U_1\otimes...\otimes U_q f[A]\ket{\psi(i_1,...,i_q)}=f[\bar{A}]\ket{i_1,...,i_q}\ket{\mbox{junk}}.
\label{st_corresp_alt}
\end{equation}
Indeed, by eq. (\ref{st_corresp}), for all tensor polynomials $g$ of degree $1$, we have that
%\begin{equation}
$
Ug[A]\ket{\psi(i_1,...,i_q)}=\sum_{\vec{j}}\bra{\vec{j}}g[\bar{A}]\ket{\vec{i}}\ket{\vec{j}}\ket{\mbox{junk}},
$
%\end{equation}
with $U=U_1\otimes...\otimes U_q$. Multiplying on both sides by $U^\dagger$, one obtains
$
%\begin{equation}
g[A]\ket{\psi(i_1,...,i_q)}=\sum_{\vec{j}}\bra{\vec{j}}g[\bar{A}]\ket{\vec{i}}\ket{\psi(\vec{j})}.
$
%\end{equation}
By induction, we have that
$
%\begin{equation}
\prod_k g_k[A]\ket{\psi(i_1,...,i_q)}=\sum_{\vec{j}}\bra{\vec{j}}\prod_k g_k[\bar{A}]\ket{\vec{i}}\ket{\psi(\vec{j})},
$
%\end{equation}
for any tuple of tensor polynomials $(g_k)_k$ of local degree $1$. The general relation (\ref{st_corresp_alt}) follows by linearity.
    
\end{remark}

These remarks make Definition~\ref{def:FST} (to our knowledge) the strongest notion of self-testing in the literature, as it certifies the underlying state and the action of the measurements \emph{component-wise} in some basis.
In particular full self-testing clearly implies eq. (\ref{st_corresp_standard}) for general (non-projective) operator systems $A$.

In Appendix~\ref{sec:qubit_selftest_lemmas} we provide a sufficient condition on the reference state $\ket{\bar{\psi}}$ and projector operators $\bar{A}$ that allows extending the standard notion of self-testing (Def. \ref{def:standard_self_testing}) to full self-testing (Def. \ref{def:FST}) for quantum Bell inequalities with binary outcomes that self-test multipartite qubit systems. More precisely,
\begin{prop}
\label{prop:dicho_FST}
    Let ${\sf B}$ be a normalized $q$-partite quantum Bell functional with dichotomic outputs in the correlator form
\begin{align}
\label{eq:corform}
{\sf B}[A]&=%\id^{(1)}\otimes {\sf D}^{(1)}_{\emptyset}[K] +
\sum_{x}K^{(1)}_{x}\otimes {\sf D}^{(1)}_{x}[K] %\nonumber\\&
=\sum_{x}K^{(2)}_{x}\otimes {\sf D}^{(2)}_{x}[K]=...%\nonumber\\&
=\sum_{x}K^{(q)}_{x}\otimes {\sf D}^{(q)}_{x}[K],
\end{align}
where each ${\sf D}^{(j)}_{x}$ denotes a tensor polynomial of the variables $\{K^{(k)}_{y}:y,k\not=j\}$, and the set of inputs $x$ is formally augmented to include $K^{(j)}_{x=\star}=\id^{(j)}$.
Suppose that ${\sf B}$ self-tests (in the standard sense, i.e., Def.~\ref{def:standard_self_testing}) the state $\ket{\bar{\psi}}\in\bigotimes_{k=1}^q H_k$ and the projectors $\bar{A}:=\{\bar{A}^{(k)}_{a|x}\in B(H_k)\}_{a,x,k}$. Further assume that the following two conditions hold:
\begin{enumerate}
    \item For all $x\not=\star$, the operator ${\sf D}^{(j)}_{x}[\bar{A}]$ is invertible.
    \item For each $k\in\{1,...,q\}$, and each $x\not=\star$,
    \begin{equation}
    {\sf{Span}}\{\bar{A}_{a|x}^{(k)}\otimes f[\{\bar{A}^{(j)}:j\not=k\}]\ket{\bar{\psi}}:a,\; f\; \mbox{degree $1$ tensor polynomial}\}=\bigotimes_{l=1}^qH_l.
    \label{gen_Hilbert_space}
    \end{equation}
\end{enumerate}
Then ${\sf B}$ fully self-tests $\ket{\bar{\psi}}$, $\bar{A}$.
\end{prop}
The proof is provided in App.~\ref{sec:qubit_selftest_lemmas}. 
%case of qubits, condition~\eqref{eq:state_selftest} is actually sufficient to imply eqs.~(\ref{state_decomp}) and~(\ref{st_corresp}), see App.~\ref{sec:qubit_selftest_lemmas}.
We next extend the definition of connector complex to refer to its self-testing properties. 
\begin{defin}
The $q\to 1$ connector complex $\CC$ is \emph{fully self-testing} if there exist coefficients $\{\mu_{b|y}:b,y\}$ such that the Bell functional
\begin{equation}
\label{eq:fst_bell_exist}
\sum_{b,y}\mu_{b|y}{\sf C}^{b|y}(\CC)[A]
\end{equation}
is a normalized quantum Bell inequality that fully self-tests $(\ket{\bar{\psi}}, \bar{A}(\CC))$, for some state $\ket{\bar{\psi}}\in I(\CC)$ \footnote{Remember that the condition $\ket{\bar{\psi}}\in I(\CC)$ is granted, for example, if $\ket{\bar{\psi}}$ is the preimage of the associate eigenvector of one of the projectors $\bar{M}_{b|y}$, cf. Sec.~\ref{sec:tight_CI}).} and, moreover,
\begin{equation}
I(\CC)=\{f[{\sf C}[\bar{A}]]\ket{\bar{\psi}}:f,\mbox{ tensor polynomial}\},
\label{cyclicity}
\end{equation}
with ${\sf C}={\sf C}(\CC)$. This condition is equivalent to $f[M(\CC)]V(\CC)\ket{\bar{\psi}}$ generating the whole range of $V(\CC)$, which we also denote as $H(\CC)$. In the following, we denote $\mu,\bar{\psi}$ by $\mu(\CC),\bar{\psi}(\CC)$.
    
\end{defin}
Of course, if a connector complex $\CC$ is fully self-testing, so will any other unitarily equivalent complex.

Remarkably, by direct application of Proposition~\ref{prop:dicho_FST},
\begin{center}
    \emph{all connectors complexes introduced in Sec.~\ref{sec:firstcomplexes} are fully self-testing}
\end{center}
(in Appendix~\ref{sec:other_complexes} we explicitly show how the associated Bell functionals comply with the conditions of Prop.~\ref{prop:dicho_FST}, as well as satisfying~\eqref{cyclicity}).

Finally, the following lemma will be instrumental to prove many-body self-testing results via the composition of full self-testing complexes.
\begin{lemma}
\label{lemma_st}
Let $\CC_1$ ($\CC_2$) be a $p\to 1$ ($q\to 1$) fully self-testing complex, and let $\DD$ be the result of congruently contracting the out-going leg of $\CC_1$ with the $k^{th}$ in-going leg of $\CC_2$. Then, $\DD$ is fully self-testing, with $\mu(\DD)=\mu(\CC_2)$ and
\begin{equation}
\ket{\bar{\psi}(\DD)}=(\bigotimes_{j<k}\id_j\otimes V^\dagger(  \CC_1)\otimes \bigotimes_{j>k}\id_j)\ket{\bar{\psi}(\CC_2)}.    
\label{final_state}
\end{equation}

\end{lemma}
The proof of this technical lemma can be found the Appendix~\ref{app:lem_prof}.

Lemma \ref{lemma_st} immediately suggests a simple method to obtain self-testing Bell inequalities, given a few fully self-testing complexes.

\begin{prop}
\label{prop:cong_contr_fst}
Let $\CC^{(n)}$ be the result of congruently contracting $n$ fully self-testing connector complexes. Then, the inequality
\begin{equation}
\langle\mu(\CC^{(n)}){\sf C}(\CC^{(n)})[A(\CC^{(n)})]\rangle\leq 1
\end{equation}
fully self-tests the local operators  $\bar{A}(\CC^{(n)})$ and the tensor network state $\ket{\bar{\psi}(\CC^{(n)})}$.
\end{prop}
\begin{proof}
Since $\CC^{(n)}$ is the result of congruently contracting fully self-testing complexes, by Lemma \ref{lemma_st}, $\CC^{(n)}$ is also fully self-testing.
\end{proof}

To conclude, the results of this section (in particular Propositions~\ref{prop:dicho_FST} and~\ref{prop:cong_contr_fst}) provide a simple recipe for the construction of arbitrarily large fully self-testing connector complexes, simply by making congruent contractions of smaller building blocks, such as the complexes described in Sec.~\ref{sec:firstcomplexes}.

\section{Examples}
\label{sec:examples}

%\subsection{Recursive CHSH tight Bell inequalities and their bounds}
\begin{figure}[h]
    \centering
    \begin{subfigure}[b]{\textwidth}
    \includegraphics[width=0.6\linewidth]{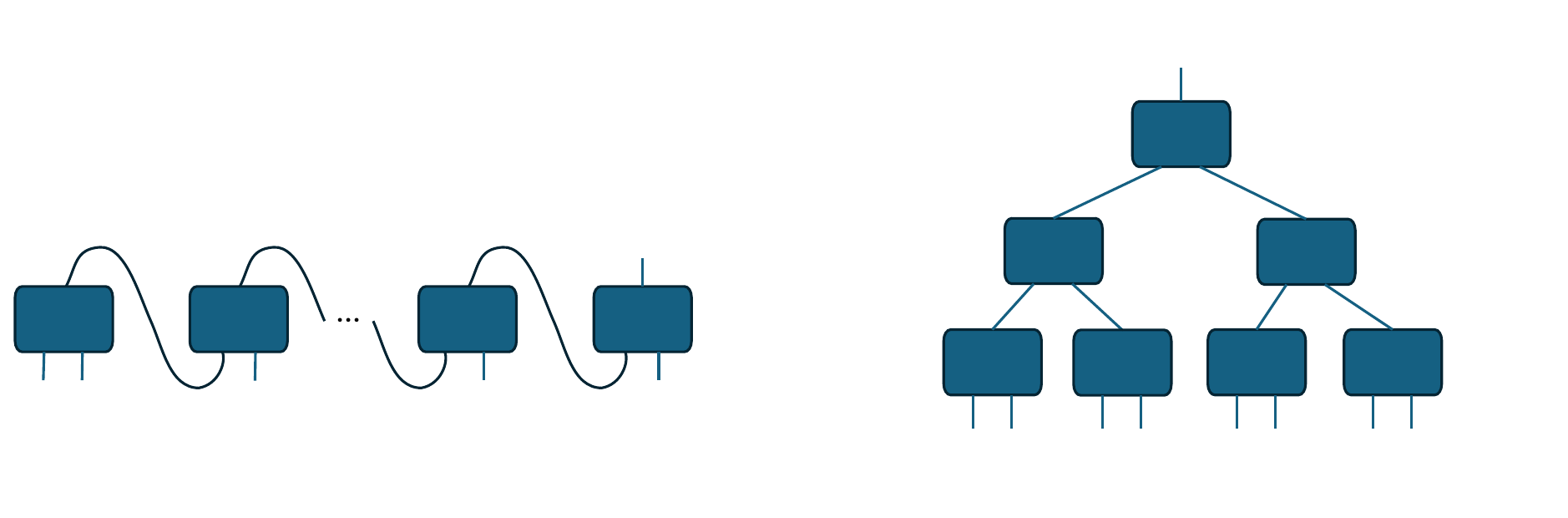}
    \caption{Simple MPS-like (left) 
    and binary-tree (right) geometries of contraction for $2\rightarrow 1$ connectors.}
    \label{fig:simp_geo_exa}
    \end{subfigure}
    \\
    \begin{subfigure}[b]{\textwidth}
    \includegraphics[width=0.6\linewidth]{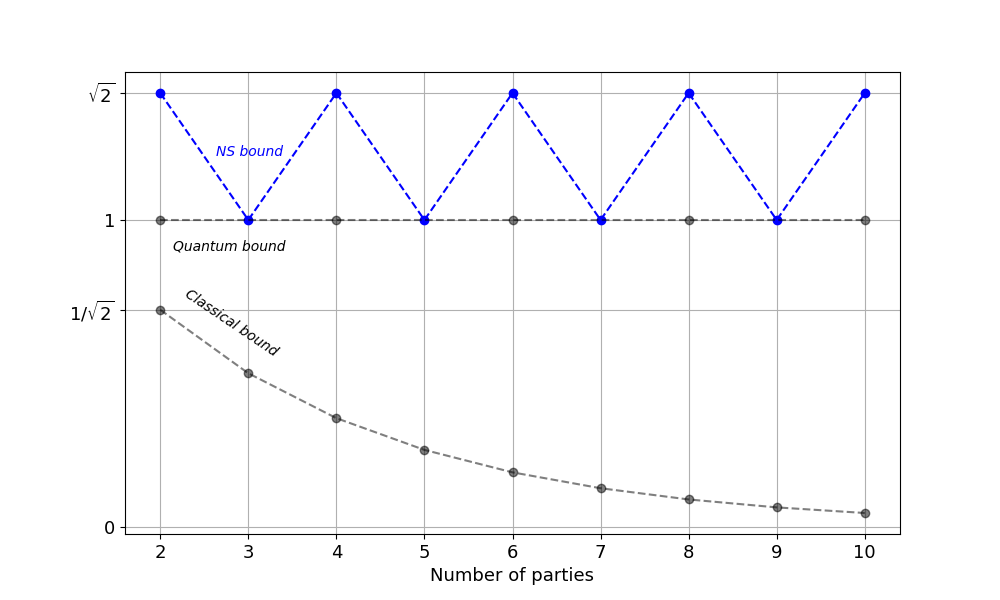} 
    \caption{Maximum values for the Bell functional resulting from a MPS-like geometry (Fig.~\ref{fig:simp_geo_exa} (left)) of Tsirelson connector complexes, which is normalized by construction to 1 for quantum behaviours. 
    The exponential decay in size of the classical bound is a property of any network based on XOR-type complexes (cf. Sec.~\ref{sec:XOR_expo}), whereas the maximum value attained by general no-signalling behaviours remains bounded by $\sqrt{2}$.}
    \label{fig:tsirel_plot}
    \end{subfigure}
    \caption{Contraction of $\CC_{\rm Tsi}$ complexes for $N$-partite fully self-testing quantum Bell inequalities.}
    \label{fig:mps_tsirelson}
\end{figure}

In his seminal paper~\cite{mermin1990extreme}, Mermin proposed a $N$-partite dichotomic Bell inequality that has an exponential quantum violation, tailored to the GHZ state. More specifically, Mermin's inequality features a $2^{(N-1)/2}$-fold quantum violation for odd $N$ and $2^{(N-2)/2}$ for even $N$.
Using a variant of the same argument, Ardehali~\cite{ardehali1992bell} proposed a similar Bell inequality with a quantum violation of factor $2^{(N-1)/2}$ for even $N$ and $2^{(N-2)/2}$ for odd $N$.
In Ardehali's paper one can glimpse the seed of a recursive argument, similar to the one exploited shortly after, in~\cite{Belinskii1993interference}, where Belinskii and Klyshko (BK) proposed a recursive variation of CHSH that features a $2^{(N-1)/2}$ violation for all $N$. It is not difficult to see that the BK inequality is the result of nesting into each other $N-1$ CHSH inequalities, corresponding ot the contraction of $N-1$ Tsirelson connectors $\mathcal{C}_{\rm Tsi}$, according to the MPS-like geometry shown in Fig.~\ref{fig:simp_geo_exa}. In Fig.~\ref{fig:tsirel_plot} we plot the corresponding maximal values of the resulting inequality increasing $N$, which is by construction normalized to $1$ for quantum behaviours. The maximum value for local behaviours corresponds to $2^{-(N-1)/2}$, as from~\cite{Belinskii1993interference} and in agreement with the results of Sec.~\ref{sec:XOR_expo}. More in general, due to CHSH being a XOR inequality (cf. Sec.~\ref{sec:XOR_expo}) with $\gamma(\mathcal{C}_{\rm Tsi})=\sqrt{2}$, the same local bound holds for any geometry of a connector network of $\mathcal{C}_{\rm Tsi}$ with a single output leg, thanks to Proposition~\ref{prop:XOR_ineq}. Interestingly, whereas the quantum-classical ratio increases exponentially with $N$, the maximal no-signalling value of the inequality can be seen (Fig.~\ref{fig:tsirel_plot}) to oscillate between $\sqrt{2}$ (PR-box violation) and $1$ (no violation).
Notice that our results and in particular Proposition~\ref{prop:cong_contr_fst} imply that the BK inequality fully self-tests the underlying optimal state and operators.
Clearly our connect-network framework can be used to generalise the inequality by contracting Tsirelson connectors $\mathcal{C}_{\rm Tsi}$ in any geometry, for example the binary tree shown in Fig.~\ref{fig:simp_geo_exa}.

\begin{figure}[h]
    \centering
    \begin{subfigure}[b]{\textwidth}
    \includegraphics[width=0.25\linewidth]{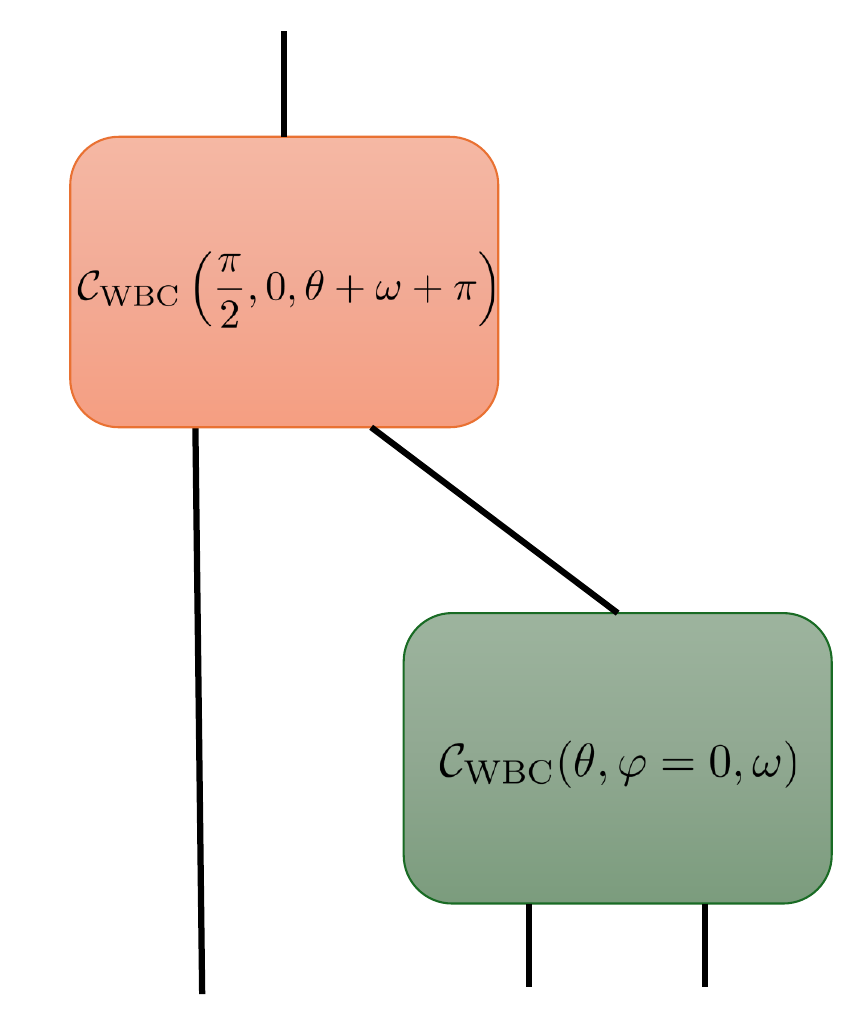}\qquad\;\;
     \includegraphics[width=0.6\linewidth]{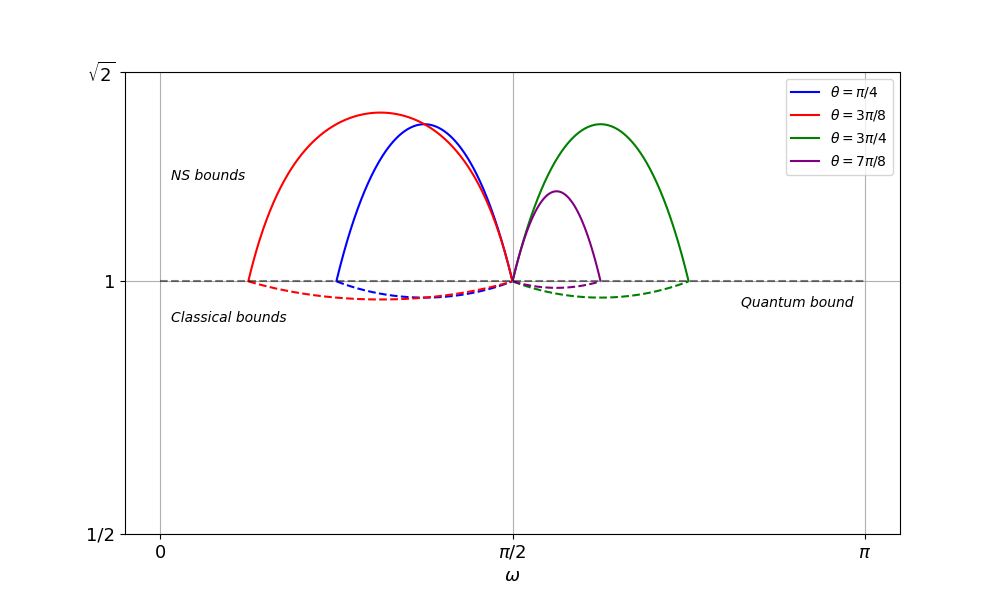}
    \caption{3-partite connector complex built via the congruent contraction of two $\CC_{\rm WBC}$.}
    \label{fig:ex_angles_3p}
    \end{subfigure}\\
    \begin{subfigure}[b]{\textwidth}
    \includegraphics[width=0.3\linewidth]{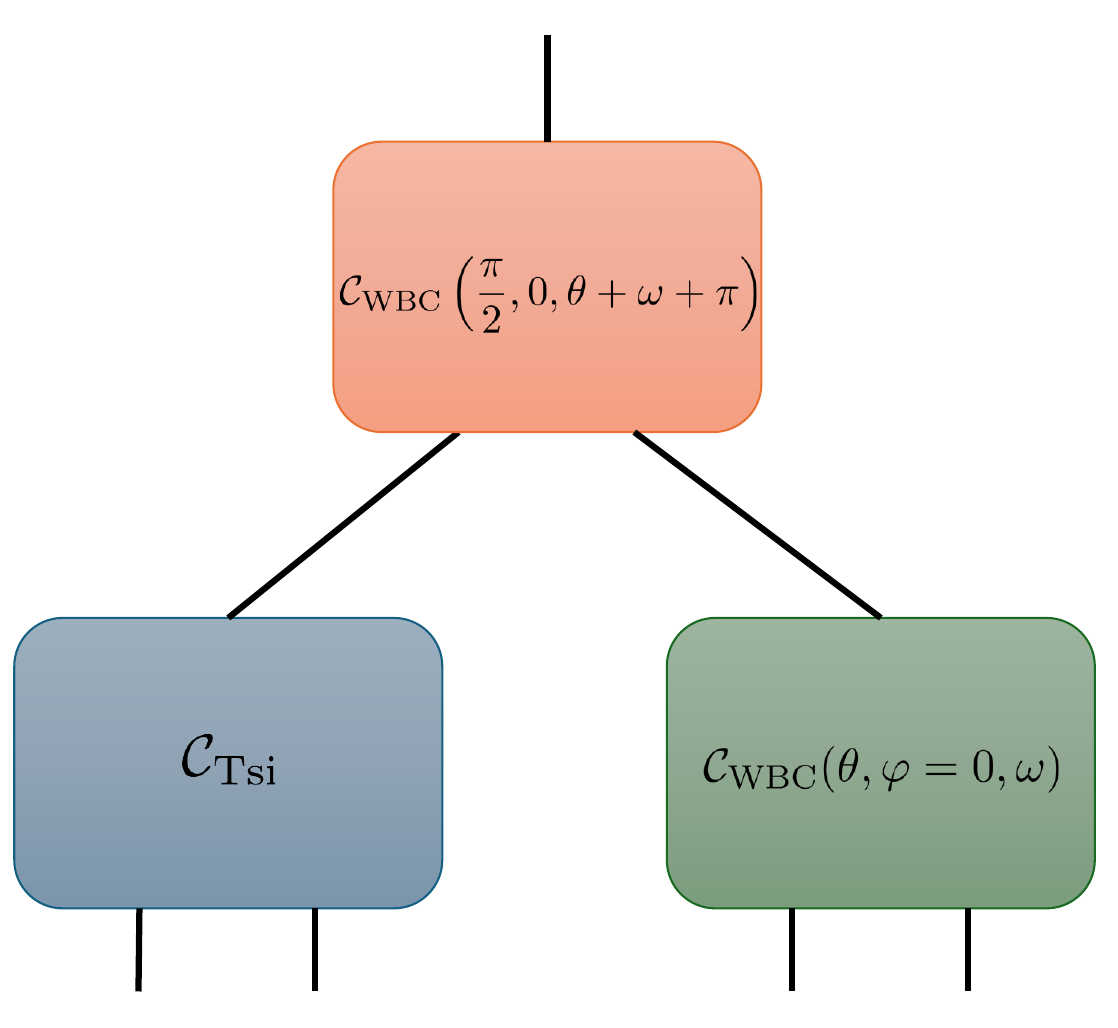}
    \includegraphics[width=0.6\linewidth]{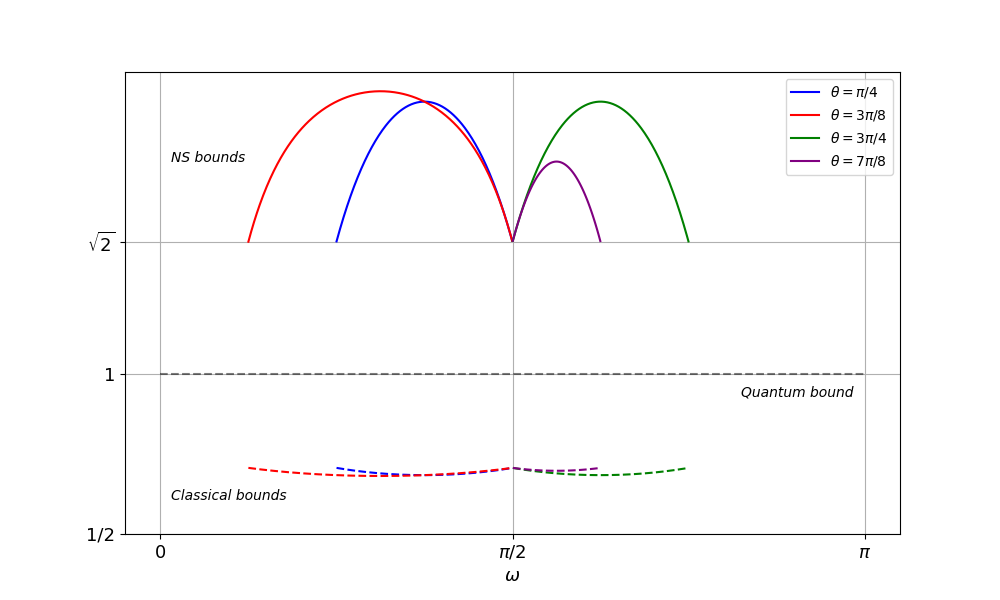}
    \caption{4-partite connector complex obtained appending an additional $\CC_{\rm Tsi}$ to the above~\ref{fig:ex_angles_3p}. The gap among classical, quantum and general no-signalling behaviours increases in this case.}
     \label{fig:ex_angles_4p}
    \end{subfigure}
    \caption{Congruent contraction of WBC complexes and corresponding maximal values on the associated inequalities. The no-signalling and classical maximal values are obtained numerically solving the corresponding linear program. For classical behaviours, this maximization yields the same bound as Proposition~\ref{prop:XOR_ineq} in both configurations.}
    \label{fig:angle_examples}
\end{figure}

As an additional application of the framework presented in this paper, we show in Figure~\ref{fig:angle_examples} numerical results for the Bell inequalities associated to the congruent contraction of $\CC_{\rm WBC}$ connectors, in a 3-partite configuration~\ref{fig:ex_angles_3p} and a similar $4$-partite configuration~\ref{fig:ex_angles_4p}  built including a $\CC_{\rm Tsi}$ connector. We plot the corresponding maximal values of the resulting Bell functional for classical and no-signalling behaviours, for different choices of parameters satisfying congruence of contraction and the consistency condition $\cos(\theta + \varphi)\cos(\varphi)\cos(\theta + \omega)\cos(\omega) < 0$ (cf.~\cite{wooltorton2023device}).
%The no-signalling and classical maximal values are obtained numerically solving the corresponding linear program. For classical behaviours, this maximization yields the same bound as Proposition~\ref{prop:XOR_ineq} in both configurations.

%One might include the quantum circuit that generates these states. It is sufficient to write the isometries $V^\dagger$ as unitaries on system+ancilla

\section{Conclusion}
In this paper, we introduced the notion of tightness in connector theory. In a nutshell, a $q\to 1$ normalized connector is tight if all its associated $q$-partite quantum Bell inequalities can be saturated with the same measurement operators. This requirement makes the congruent contraction of any number of tight connectors result in another tight connector, from which one can derive new, multipartite tight quantum Bell inequalities. Moreover, such inequalities, which are maximally violated by tensor network states \cite{tensor_network_states}, can be shown self-testing if each of the tight connectors involved satisfies an extra, verifiable property, which we dubbed full self-testing.

To our surprise, instances of tight, fully self-testing connectors can be constructed from the oldest quantum Bell inequality: Tsirelson's bound \cite{tsirelson_bound}. As we showed, recently discovered self-testing quantum Bell inequalities \cite{acin2012randomness, mckague2014self, wooltorton2023device} can be similarly used to build families of tight, fully self-testing connectors. Moreover, one such family generates, by contraction, XOR Bell functionals for which the ratio between the maximum quantum and classical values increases exponentially with the system size. All in all, the identified connectors allow one to construct analytic families of quantum Bell inequalities with which one can self-test arbitrary tree tensor network states of both bond and local dimension $2$, as well as arbitrary pairs of local qubit projectors.

For future research, it would be interesting to extend our results to unnormalized connectors, which in \cite{connectors} were proven advantageous to detect certain types of supra-quantum nonlocality. A similar generalization of our concepts to $q\to p$ connectors would provide us with new tools to construct multi-partite quantum Bell inequalities, although it is worth noting that $q\to 2$ connectors that do not reduce to the contraction of several $p\to 1$ connectors (or convex combinations thereof) are not even known to exist \cite{connectors}. Finally, now that we have large families of tight connectors, the next logical step is exploring their potential to generate \emph{local} quantum Bell inequalities, in the manner sketched in \cite{Kull_2024}. 

\section*{Acknowledgements}
\begin{minipage}[l]{0.8\textwidth}
    This project was funded within the QuantERA II Programme that has received funding from the European Union's Horizon 2020 research and innovation programme under Grant Agreement No 101017733, and from the Austrian Science Fund (FWF), projects I-6004 and ESP2889224 as well as Grant No. I 5384.
\end{minipage}
\qquad
\begin{minipage}[r]{0.2\textwidth}
    \includegraphics[width=2cm]{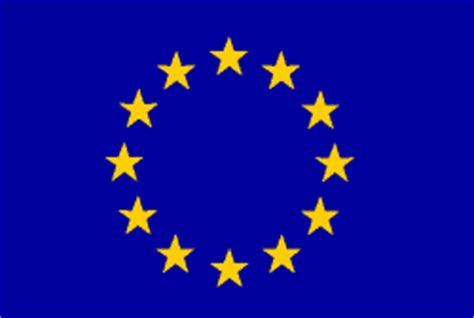} 
\end{minipage}

%\section*{Acknowledgements}
%\begin{wrapfigure}{r}[0cm]{2cm}
%\begin{center}
%\includegraphics[width=2cm]{EU_logo}    
%\end{center}
%\end{wrapfigure} 
%This project was funded within the QuantERA II Programme that has received funding from the European Union's Horizon 2020 research and innovation programme under Grant Agreement No 101017733, and from the Austrian Science Fund (FWF), project I-6004 and project ESP2889224. 

\bibliography{references}
\newpage

\begin{appendix}
\section{The invariant subspace of a connector complex}
\label{app:invariant}
Let $\Lambda$ satisfy $\id\geq\Lambda\geq 0$, let $V$ be a co-isometry such that $\Pi:=V\Lambda V^\dagger$ is a projector and call $I$ the range of $V^\dagger$. We next prove that $\Lambda I\subset I$.

Let $\{\ket{\phi^0_k}\}_k$, $\{\ket{\phi^1_k}\}_k$ respectively represent bases for the eigenvectors of $\Pi$ with eigenvalues $0$ and $1$. Then we have that $I={\sf{Span}}\{V^\dagger\ket{\phi^j_k}:j=0,1, k\}$. Now, for arbitrary $k$ we have that
\begin{equation}
\bra{\phi^0_k}V\Lambda V^\dagger\ket{\phi^0_k}=\bra{\phi^0_k}\Pi\ket{\phi^0_k}=0.
\end{equation}
Since $\Lambda\geq 0$, this implies that $\Lambda V^\dagger\ket{\phi^0_k}=\sqrt{\Lambda}\sqrt{\Lambda} V^\dagger\ket{\phi^0_k}=0\in I$, for all $k$. Similarly, from $\id-\Lambda\geq 0$ and the relation
\begin{equation}
\bra{\phi^1_k}V(\id-\Lambda) V^\dagger\ket{\phi^1_k}=0,
\end{equation}
it follows that $\Lambda V^\dagger\ket{\phi^1_k}=V^\dagger\ket{\phi^1_k}\in I$, for all $k$. Thus, $\Lambda V^\dagger \ket{\phi^j_k}\in I$, for all $j,k$ and so $\Lambda I\subset I$.

\section{Connector networks and their corresponding tensor network states}
\label{app: Tensor network state}
In this appendix we clarify the duality between any network of tight connectors and its corresponding optimal state $\ket{\psi}$, which turns up to be a tensor network state according to the same network topology, as outlined in the main text Sec.~\ref{sec:tight_CI}.

To start, consider for simplicity the case of two connectors, a $p\rightarrow 1$ connector $\CC_1$ and a $q\rightarrow 1$ connector $\CC_2$ as represented in Fig.~\ref{fig:basic_contraction}.
%gives more insights into the contraction of local isometries, which give rise to a tensor network state, .
%The tree network depicted there consists of six congruently contracted connector complexes $\CC_1, ..., \CC_6$, each containing a $2 \rightarrow 1$ connector. However, for the sake of generality consider instead $q \rightarrow 1$ connectors as well as a more sophisticated network consisting of $\CC_1, ..., \CC_N$ connector cmplexes.

As detailed in the main text the co-isometry of the resulting global connector is given by
\begin{equation}
V=V(\CC_2)\left(\bigotimes_{j<k}\id_j\otimes V(\CC_1)\otimes \bigotimes_{j>k}\id_j\right),
\end{equation}
and the optimal state $\ket{\psi}=V^\dagger\ket{\phi}$, where $\phi$ is an eigenvector of $M(\CC_2)$.
Therefore
\begin{align}
    \ket{\psi}=\left(\bigotimes_{j<k}\id_j\otimes V^\dagger(\CC_1)\otimes \bigotimes_{j>k}\id_j\right)V^\dagger(\CC_2)\ket{\phi}
\end{align}
which, in the computational basis, reads:
\begin{align}
    \ket{\psi}=\sum_{i_1,\dots,i_{k-1},j_1,\dots,j_p,i_{k+1},\dots,i_{q}}
    \sum_{i_k,l} V^\dagger(\CC_1)^{i_k}_{j_1,\dots,j_p}V^\dagger(\CC_2)^{l}_{i_1,\dots,i_q}\phi_l\ket{i_1,\dots,i_{k-1},j_1,\dots,j_p,i_{k+1},\dots,i_{q}}\;.
\end{align}
By definition, this is the tensor network state arising from the corresponding contraction of the two tensors $ V^\dagger(\CC_1)^{i_k}_{j_1,\dots,j_p} $ and $ \sum_{l} V^\dagger(\CC_2)^{l}_{i_1,\dots,i_q}\phi_l$.

Clearly, this construction can be generalized to any network, resulting in the network state represented by Fig.~\ref{fig:conn2tree}

\paragraph*{Matrix product states.}
As an example, it is insightful to consider the network depicted in Fig.~\ref{fig:simp_geo_exa}(left), corresponding to various tensors that are sequentially fed into each other. It is not difficult to realize that in such case the optimal quantum state is of the form 
\begin{align}
 \ket{\psi}=\sum_{i_1,...,i_{m+1}}\sum_{o_1,\dots,o_m} 
 V^\dagger(\CC_1)^{o_1}_{i_{1},i_{2}} V^\dagger(\CC_2)^{o_2}_{o_1,i_{3}}
 \dots 
 V^\dagger(\CC_{m})^{o_{m}}_{o_{m-1},i_{m+1}}\phi_{o_m}\ket{i_1,...,i_{m+1}},
\end{align}
which, in standard tensor network notation \cite{MPS}, corresponds to a matrix product state (MPS) generated by the matrices
\begin{align}
&T^{[0]}_i:=\bra{i},\nonumber\\
&(T^{[k]}_i)_{\alpha,\beta}:=(\bra{\alpha}\otimes \bra{i})V^\dagger(\CC_k)\ket{\beta}, \mbox{ for } k=1,...,m-1,\nonumber\\
&T^{[m+1]}_i:=(\id\otimes \bra{i})V^\dagger(\CC_{m})\ket{\phi}.
\end{align}

\section{Promoting standard to full self-testing results}
\label{sec:qubit_selftest_lemmas}
In this Appendix, we will prove Proposition \ref{prop:dicho_FST}, which provides sufficient conditions to extend a standard self-testing claim to a full self-testing result.

% will prove results that are useful to simplify full self-testing (Def.~\ref{def:FST}) proofs, reducing it, for the case of qubits systems, to the notion of standard self-testing (Def.~\ref{def:standard_self_testing}), that is the validity of~\eqref{st_corresp_standard} under the assumption of the implemented measurements $A$ being projective.

% The proof is divided in two lemmas that, when combined, prove our Proposition~\ref{prop:dicho_FST}, namely:
% \begin{itemize}
%     \item first, Lemma~\ref{lem:state_to_compo-wise} proves that condition~\eqref{st_corresp_standard} implies eqs.~\eqref{state_decomp} and~\eqref{st_corresp} for qubits systems satisying a simple condition (eq.~\ref{eq:ci_cond}).
%     \item secondly, under the same hypothesis Lemma~\ref{lem:povm_to_proj_QB}
%     proves that usual correlator Bell functionals that are self-testing in the standard sense, are fully self-testing.
% \end{itemize}

To do so, we will need the following lemma.
\begin{lemma}
\label{lem:state_to_compo-wise}
Let $\bar{A}$, be a $q$-partite measurement system of projective dichotomic measurements, and let $\ket{\bar{\psi}}\in \bigotimes_{j=1}^qH_j$ be a normalized quantum state, such that condition (\ref{gen_Hilbert_space}) is satisfied for all $x,k$. Let $U_1,...,U_q$ be local isometries and suppose that, for some (not necessarily projective) $q$-partite measurement system $A$ and some $q$-partite state $\ket{\psi}$, the identity
\begin{equation}
\label{eqapp:state_selftest}
    U_1\otimes...\otimes U_q \bigotimes_k A^{(k)}_{a_k|x_k}\ket{\psi}=\bigotimes_k \bar{A}^{(k)}_{a_k|x_k}\ket{\bar{\psi}}\ket{\mbox{junk}},
\end{equation}
holds for all $a_1,...,a_q,x_1,...,x_q$. 

Then, there exists an orthonormal set of states $\{\ket{\psi(i_1,...,i_q)}:\vec{i}\}$ such that eqs.~\eqref{state_decomp} and~\eqref{st_corresp} are satisfied.
\end{lemma}

\begin{proof}
In the following, for ease of reading we will sometimes drop the tensor products when this does not lead to misunderstandings. 

For $x,k$ we have that
\begin{align}
&\bra{\psi}(A^{(k)}_{a|x})^2\ket{\psi}=\bra{\psi}A^{(k)}_{a|x}U^\dagger UA^{(k)}_{a|x}\ket{\psi}\nonumber\\
&\overset{(\ref{eqapp:state_selftest})}{=}\bra{\bar{\psi}}(\bar{A}_{a|x}^{(k)})^2\ket{\bar{\psi}}=\bra{\bar{\psi}}\bar{A}_{a|x}^{(k)}\ket{\bar{\psi}}\nonumber\\
&\overset{(\ref{eqapp:state_selftest})}{=}\bra{\psi}A^{(k)}_{a|x}\ket{\psi},
\label{null_cond}
\end{align}
where we have invoked that $\bar{A}$ is projective, i.e., $(\bar{A}_{a|x}^{(k)})^2=\bar{A}_{a|x}^{(k)}$.

Define $Y:={A}_{a|x}^{(k)}-({A}_{a|x}^{(k)})^2$. This operator is positive semidefinite, since $0\leq {A}_{a|x}^{(k)}\leq \id_k$. Eq. (\ref{null_cond}) implies that $\bra{\psi}Y\ket{\psi}=0$, and so $Y\ket{\psi}=0$, or equivalently,
\begin{equation}
(A^{(k)}_{a|x})^2\ket{\psi}=A^{(k)}_{a|x}\ket{\psi}.
\label{eq:A0A0}
\end{equation}
Let $\bar{a}$ denote the complement of $a$. In view of the equation above, we have that
\begin{equation}
A^{(k)}_{a|x}\ket{\psi}=A^{(k)}_{a|x}(A^{(k)}_{a|x}+A^{(k)}_{\bar{a}|x})\ket{\psi}=A^{(k)}_{a|x}\ket{\psi}+A^{(k)}_{a|x}A^{(k)}_{\bar{a}|x}\ket{\psi},
\end{equation}
from which we infer that
\begin{equation}
A^{(k)}_{a|x}A^{(k)}_{\bar{a}|x}\ket{\psi}=0.
\label{eq:A0A1}
\end{equation}

Let $\bar{A}^{(\not=k)}$ ($A^{(\not=k)}$) denote the corresponding measurement system for parties $1,...,k-1,k+1,...,q$, that is, $\{\bar{A}^{(j)}:j\not=k\}$ ($\{A^{(j)}:j\not=k\}$). By condition (\ref{gen_Hilbert_space}), for any $\vec{i}, x,k$ there exist tensor polynomials $f^0_{\vec{i},x,k}, f^1_{\vec{i},x,k}$ of local degree $1$ such that
\begin{equation}
\ket{\vec{i}}=\left(\bar{A}^{(k)}_{0|x}f^0_{\vec{i},x,k}[\bar{A}^{(\not=k)}]+\bar{A}^{(k)}_{1|x}f^1_{\vec{i},x,k}[\bar{A}^{(\not=k)}]\right)\ket{\bar{\psi}}.
\label{identity_vectis}
\end{equation}
Note that the term between brackets is a degree $1$ tensor polynomial.
Define thus the vector
\begin{equation}
\ket{\psi_{x,k}(\vec{i})}:=\left(A^{(k)}_{0|x}f^0_{\vec{i},x,k}(A^{(\not=k)})+A^{(k)}_{1|x}f^1_{\vec{i},x,k}(A^{(\not=k)})\right)\ket{\psi}.
\label{def_psi_xk}
\end{equation}
It can be seen that $\ket{\psi_{x,k}(\vec{i})}$ is independent of $x,k$: just act with the operator $U:=\bigotimes_{j=1}^qU_j$ on both sides of the definition. By eqs. (\ref{identity_vectis}) and (\ref{eqapp:state_selftest}), we find that
\begin{equation}
U\ket{\psi_{x,k}(\vec{i})}=\ket{\vec{i}}\ket{\mbox{junk}}.
\end{equation}
Since $U$ is invertible, it follows that $\ket{\psi_{x,k}(\vec{i})}=\ket{\psi_{x',k'}(\vec{i})}$. Thus, from now on, we just use denote the vector by $\ket{\psi(\vec{i})}$.

From the relation 
\begin{equation}
U\ket{\psi(\vec{i})}=\ket{\vec{i}}\ket{\mbox{junk}},
\label{bare_basis_rel}
\end{equation}
it follows that the vectors $\{\ket{\psi(\vec{i})}:\vec{i}\}$ are normalized and orthogonal to each other. Moreover, acting with $U$ on the vector
\begin{equation}
\sum_{\vec{i}}\braket{\vec{i}}{\bar{\psi}}\ket{\psi(\vec{i})}-\ket{\psi}
\end{equation}
we obtain, by Eqs.~(\ref{eqapp:state_selftest}), (\ref{bare_basis_rel}) the null vector. Since $U$ is invertible, eq.~(\ref{state_decomp}) follows.

It remains to prove eq. (\ref{st_corresp}). Let $\vec{i},x,k$ be arbitrary. Then,
\begin{align}
&A^{(k)}_{a|x}\ket{\psi(\vec{i})}\nonumber\\
&\overset{(\ref{def_psi_xk})}{=}A^{(k)}_{a|x}\left(A^{(k)}_{0|x}f^0_{\vec{i},x,k}[A^{(\not=k)}]+A^{(k)}_{1|x}f^1_{\vec{i},x,k}[A^{(\not=k)}]\right)\ket{\psi}\nonumber\\
&\overset{(\ref{eq:A0A0}, \ref{eq:A0A1})}{=}\left(\delta_{a,0}A^{(k)}_{0|x}f^0_{\vec{i},x,k}[A^{(\not=k)}]+\delta_{a,1}A^{(k)}_{1|x}f^1_{\vec{i},x,k}[A^{(\not=k)}]\right)\ket{\psi}.
\end{align}
The expression between brackets in the last line is a tensor polynomial of degree $1$, so we can multiply the identity by $U$ and invoke eq. (\ref{eqapp:state_selftest}) to obtain
\begin{align}
&UA^{(k)}_{a|x}\ket{\psi(\vec{i})}\nonumber\\
&=\left(\delta_{a,0}\bar{A}^{(k)}_{0|x}f^0_{\vec{i},x,k}(\bar{A}^{(\not=k)})+\delta_{a,1}\bar{A}^{(k)}_{1|x}f^1_{\vec{i},x,k}(\bar{A}^{(\not=k)})\right)\ket{\bar{\psi}}\ket{\mbox{junk}}\nonumber\\
&=\bar{A}^{(k)}_{a|x}\left(\bar{A}^{(k)}_{0|x}f^0_{\vec{i},x,k}(\bar{A}^{(\not=k)})+\bar{A}^{(k)}_{1|x}f^1_{\vec{i},x,k}(\bar{A}^{(\not=k)})\right)\ket{\bar{\psi}}\ket{\mbox{junk}}\nonumber\\
&=\bar{A}^{(k)}_{a|x}\ket{\vec{i}}\ket{\mbox{junk}}.
\end{align}
This last relation can also be written as
\begin{align}
&UA^{(k)}_{a|x}\ket{\psi(\vec{i})}\nonumber\\
&=\sum_{j_k}\bra{j_k}\bar{A}^{(k)}_{a|x}\ket{i_k}\ket{\vec{i}^-,j_k,\vec{i}^+}\ket{\mbox{junk}},
\end{align}
with $\vec{i}^-:=(i_1,...,i_{k-1})$, $\vec{i}^+:=(i_{k+1},...,i_q)$. Multiplying by $U^\dagger$, we have that
\begin{align}
&A^{(k)}_{a|x}\ket{\psi(\vec{i})}\nonumber\\
&=\sum_{j_k}\bra{j_k}\bar{A}^{(k)}_{a|x}\ket{i_k}\ket{\psi(\vec{i}^-,j,\vec{i}^+)}.
\end{align}

Now, let $x_1,...,x_q$, $a_1,...,a_q$ be arbitrary. The last relation implies, by induction, that
\begin{equation}
\bigotimes_{k=1}^qA_{a_k|x_k}\ket{\psi(i_1,...,i_q)}=\sum_{\vec{j}}\prod_{k=1}^q\bra{j_k}\bar{A}^{(k)}_{a_k|x_k}\ket{i_k}\ket{\psi(\vec{j})}.
\end{equation}
Applying $U$ and taking into account equation (\ref{bare_basis_rel}), we have that
\begin{equation}
U\bigotimes_{k=1}^qA_{a_k|x_k}\ket{\psi(i_1,...,i_q)}=\sum_{\vec{j}}\prod_{k=1}^q\bra{j_k}\bar{A}^{(k)}_{a_k|x_k}\ket{i_k}\ket{\vec{j}}\ket{\mbox{junk}}=\bigotimes_{k=1}^q\bar{A}_{a_k|x_k}\ket{\vec{i}}\ket{\mbox{junk}}.
\end{equation}
Hence, eq. (\ref{st_corresp}) also holds and the lemma is proven.

\end{proof}

\begin{proof}[Proof of Proposition \ref{prop:dicho_FST}]
Let $\ket{\psi}$, $\{K^{(k)}_{x}\}$ (we switch to the correlator basis) saturate the inequality ${\sf B}(P)\leq 1$. Since all extreme dichotomic POVMs are projective \cite[Proposition 2.41]{Heinosaari_Ziman_2011}, each operator $K^{(k)}_{x}$ can be expressed as a convex combination of ``projective'' dichotomic operators $\{K^{(k)}_{x,j}\}_j$. That is, for some probability distribution $p(j|x,k)$ with full support, it holds that
\begin{equation}
K^{(k)}_{x}=\sum_jp(j|x,k)K^{(k)}_{x,j}, 
\end{equation}
where $(K^{(k)}_{x,j})^2=1$, for all $j$. Let $J$ be a function assigning an index $j$ to every pair $(x,k)$, and use $K(J)$ to denote the set of projective dichotomic operators $\{K^{(k)}_{x,J(k,x)}\}$. Then, by convexity we have that
\begin{equation}
\bra{\psi}{\sf B}[K(J)]\ket{\psi}=1,
\end{equation}
for all functions $J$. Since ${\sf B}[K(J)]\leq \id$, it follows that
\begin{equation}
{\sf B}[K(J)]\ket{\psi}=\ket{\psi}.
\end{equation}
Take one such function $J^0$. Note that all the above implies, for $k=1,...,q$, that
\begin{equation}
\sum_{x}K^{(k)}_{x,j_x}\otimes {\sf D}^{(k)}_{x}[K(J^0)]\ket{\psi}=\ket{\psi},
\end{equation}
holds for all $(j_x)_x$. As each $j_x$ can be independently chosen, the vector
\begin{equation}
K^{(k)}_{x,j}\otimes {\sf D}^{(k)}_{x}[K(J^0)]\ket{\psi}    
\label{vect_indep}
\end{equation}
does not depend on $j$. 

Since ${\sf B}$ is self-testing for projective measurements, we can invoke Lemma~\ref{lem:state_to_compo-wise}, with the state $\ket{\psi}$ and the dichotomic operators $K(J^0)$, concluding that eqs. (\ref{state_decomp}), (\ref{st_corresp}) hold for $K(J^0)$, for some local isometry $U=\bigotimes_jU_j$. By Remark \ref{remark:st_tensor_pol}, for any tensor polynomial $f$, it is satisfied that
\begin{equation}
Uf[K(J^0)]\ket{\psi}=f[\bar{K}]\ket{\bar{\psi}}\ket{\mbox{junk}}.
\end{equation}

By assumption, the operator ${\sf D}^{(k)}_{x}[\bar{K}]$ is invertible. Thus, there exists a one-variable polynomial $g$ such that 
\begin{equation}
g[{\sf D}^{(k)}_{x}[\bar{K}]]{\sf D}^{(k)}_{x}[\bar{K}]=1.    
\end{equation}
Consequently,
\begin{align}
&g[{\sf D}^{(k)}_{x}[K(J^0)]]{\sf D}^{(k)}_{x}[K(J^0)])\ket{\psi}\nonumber\\
&=U^\dagger Ug[{\sf D}^{(k)}_{x}[K(J^0)]]{\sf D}^{(k)}_{x}[K(J^0)]\ket{\psi}\nonumber\\
&=U^\dagger g[{\sf D}^{(k)}_{x}[\bar{K}]]{\sf D}^{(k)}_{x}[\bar{K}]\ket{\bar{\psi}}\ket{\mbox{junk}}\nonumber\\
&=U^\dagger \ket{\bar{\psi}}\ket{\mbox{junk}}\nonumber\\
&=U^\dagger U\ket{\psi}\nonumber\\
&=\ket{\psi}.
\end{align}
Now, multiply eq. (\ref{vect_indep}) by $g[{\sf D}^{(k)}_{x}[K(J^0)]]$ (note that this operator does not depend on $j$ and commutes with $K^{(k)}_{x,j}$). The result, by the equation above, is that the vector
\begin{equation}
K^{(k)}_{x,j}\ket{\psi}    
\end{equation}
does not depend on $j$ either.
Thus, by convexity, 
\begin{equation}
K^{(k)}_{x}\ket{\psi}=K^{(k)}_{x,j}\ket{\psi},
\end{equation}
for all $j$, or, equivalently,
\begin{equation}
A^{(k)}_{a|x}\ket{\psi}=A^{(k)}_{a|x,j}\ket{\psi},
\end{equation}
for all $a,x,k,j$.

It follows that
\begin{align}
&U\bigotimes_k A^{(k)}_{a_k|x_k}\ket{\psi}\nonumber\\
&=U\bigotimes_k A^{(k)}_{a_k|x_k,J^0(x_k,k)}\ket{\psi}\nonumber\\
&=\bigotimes_k \bar{A}^{(k)}_{a_k|x_k}\ket{\bar{\psi}}\ket{\mbox{junk}}.
\end{align}
Now, we invoke the Lemma~\ref{lem:state_to_compo-wise} again, concluding that relations (\ref{state_decomp}), (\ref{st_corresp}) hold for $A$. Property (\ref{reconst_prop}) is implied by eq. (\ref{gen_Hilbert_space}), so we conclude that ${\sf B}$ fully self tests $\bar{A},\ket{\bar{\psi}}$.

\end{proof}

\section{Proof of Lemma~\ref{lemma_st}}
\label{app:lem_prof}
In this appendix, we provide the formal proof of Lemma~\ref{lemma_st} of the main text.

\begin{proof}
For ease of notation, we here respectively denote $A(\CC_1), \bar{A}(\CC_1)$ ($A(\CC_2)$, $\bar{A}(\CC_2)$) by $A, \bar{A}$ ($B,\bar{B}$). Similarly, $\mu$ and ${\sf C}$ will denote $\mu(\CC_1)$ and ${\sf C}(\CC_1)$. We also define the states 
\begin{align}
&\ket{\bar{\phi}}:=\ket{\bar{\psi}(\CC_1)},\nonumber\\
&\ket{\check{\phi}}:=V(\CC_1)\ket{\bar{\psi}(\CC_1)},
\label{def_phises}
\end{align}
Finally, the symbol $B^0$ ($\bar{B}^0$) will represent the measurement system variables (operators) $\{B^{(j)}:j\not= k\}$ ($\{\bar{B}^{(j)}:j\not= k\}$).

Now, to the proof. First of all, given the state $\ket{\bar{\psi}(\DD)}$ defined by eq. (\ref{final_state}) and the operators 
\begin{equation}
\bar{D}=\bar{B}^{(1)},...,\bar{B}^{(k-1)},\bar{A}^{(1)},...,\bar{A}^{(p)},\bar{B}^{(k+1)},...,\bar{B}^{(q)},    
\end{equation}
it is easy to see that property (\ref{reconst_prop}) holds. Indeed, by the fully self-testing property of $\CC_2$, we have that, for every $i_1,...,i_{k-1},i_{k+1},...i_q$, there exists a tensor polynomial $g[B]$ such that
\begin{equation}
g[\bar{B}]\ket{\bar{\psi}(\CC_2)}=\ket{i_1,...,i_{k-1}}\ket{\check{\phi}}\ket{i_{k+1},...,i_{q}}.
\end{equation}
This implies that
\begin{align}
&g[\bar{B}^0,{\sf C}[\bar{A}]]\ket{\bar{\psi}(\DD)}\nonumber\\
&=V(\CC_1)^\dagger g[\bar{B}^0,\bar{B}^{(k)}]\ket{\bar{\psi}(\CC_2)}\nonumber\\
&=\ket{i_1,...,i_{k-1}}\ket{\bar{\psi}(\CC_1)}\ket{i_{k+1},...,i_q}\;,
\end{align}
where we used the fact that the co-isometry $V(\CC_1)$ is an isometry when restricted to $I(\CC_1)$, i.e. $V^\dagger(\CC_1)V(\CC_1)|_{I(\CC_1)}=\id|_{I(\CC_1)}$.
By the fully self-testing property of $\CC_1$, we thus have that, for all $l_1,...,l_p$, there exist $h[A]$ such that
\begin{align}
&h[\bar{A}]g[\bar{B}^0,{\sf C}[\bar{A}]]\ket{\bar{\psi}(\DD)}\nonumber\\
&=\ket{i_1,...,i_{k-1}}\ket{l_1,...,l_p}\ket{i_{k+1},...,i_q}.
\end{align}
Property (\ref{reconst_prop}) thus holds for the connector complex $\DD$.

For $\DD$, property (\ref{cyclicity}) is equivalent to the statement that $\{f[\bar{M}(\DD)]\ket{\alpha}\}$ generates $H(\DD)$, with $\ket{\bar{\psi}(\DD)}=V(\DD)^\dagger\ket{\alpha}$. Since $\bar{M}(\DD)=\bar{M}(\CC_2)$ and $\ket{\bar{\psi}(\CC_2)}=V(\CC_2)^\dagger\ket{\alpha}$, the statement follows from the assumption that $\CC_2$ is fully self-testing.

We still need to prove that conditions (\ref{state_decomp}) and (\ref{st_corresp}) hold true. Consider, on the  $p+q-1$-partite system of measurement operators $D=((B^{(j)})_{j<k}, A, (B^{(j)})_{j>k})$, the normalized Bell functional ${\sf D}=\mu(\CC_2)({\sf C}(\DD))$. Since $\DD$ is the result of congruently contracting $\CC_2$ and $\CC_1$, the maximum quantum value of ${\sf D}$ is $1$. Now, let there exist a state $\ket{\psi}$ and a $p+q-1$-partite system of measurement operators $D=((B^{(j)})_{j<k}, A, (B^{(j)})_{j>k})$ such that
\begin{equation}
\bra{\psi}{\sf D}(D)\ket{\psi}=1.
\end{equation}
We have that ${\sf C}(\DD)[D]={\sf C}(\CC_2)[B]$, with $B=((B^{(j)})_{j<k}, {\sf C}[A], (B^{(j)})_{j>k})$. Since ${\sf C}$ is a connector, $B^{(k)}:={\sf C}[A]$ defines a system of one-partite measurement operators. Thus, the quantum realization $(\psi,B)$ achieves the maximum value $1$ on the functional $\mu(\CC_2)({\sf C}(\CC_2)[B])$ (where we now regard $\ket{\psi}$ as $q$-partite by relabelling the $k$th Hilbert space).

Since $\CC_2$ is completely self-testing, there exist an orthonormal set of states $\{\ket{\psi(i_1,...,i_q)}:\vec{i}\}$ and isometries $U_1,...,U_q$ such that 
\begin{equation}
\ket{\psi}=\sum_{\vec{i}}\braket{\vec{i}}{\bar{\psi}(\CC_2)}\ket{\psi(\vec{i})}
\label{state_B}
\end{equation}
and
\begin{equation}
U_1\otimes...\otimes U_qf[B]\ket{\psi(i_1,...,i_q)}=f[\bar{B}]\ket{i_1,...,i_q}\ket{\mbox{junk}}
\label{st_corresp_b}
\end{equation}
holds for all tensor polynomials $f$. For $j=1,...,\mbox{dim}(H_k(\CC_2))$, define the set of orthonormal vectors:
\begin{equation}
\ket{\psi^{(k)}_j}:=U^\dagger_k\ket{j}_k\ket{\mbox{junk}}.
\label{def_psi_k}
\end{equation}
With this definition, it holds that
\begin{align}
&U_kg[B^{(k)}]\ket{\psi^{(k)}_j}=g[\bar{B}^{(k)}]\ket{j}\ket{\mbox{junk}}.
\label{local_k_ops}
\end{align}
Indeed, by definition the left-hand side equals
\begin{align}
&U_kg[B^{(k)}]U^\dagger_k\ket{j}\ket{\mbox{junk}}\nonumber\\
&=\left(\bra{1}^{\otimes (k-1)}\otimes\id_k\otimes \bra{1}^{\otimes (q-k)}\right)U_kg[B^{(k)})]U_k^\dagger\ket{1}^{\otimes (k-1)}\ket{j}\ket{1}^{\otimes (q-k)}\ket{\mbox{junk}}\nonumber\\
&=\left(\bra{1}^{\otimes (k-1)}\otimes\id_k\otimes \bra{1}^{\otimes (q-k)}\right)U_kg[B^{(k)}]U_k^\dagger U\ket{\psi(1,...,1,j,1,...,1)}\ket{\mbox{junk}}\nonumber\\
&=\left(\bra{1}^{\otimes (k-1)}\otimes\id_k\otimes \bra{1}^{\otimes (q-k)}\right)Ug[B^{(k)}]\ket{\psi(1,...,1,j,1,...,1)}\nonumber\\
&\overset{(\ref{st_corresp_b})}{=}\left(\bra{1}^{\otimes (k-1)}\otimes\id_k\otimes \bra{1}^{\otimes (q-k)}\right)g[\bar{B}^{(k)}]\ket{1}^{\otimes (k-1)}\ket{j}\ket{1}^{\otimes (q-k)}\ket{\mbox{junk}}\nonumber\\
&=g[\bar{B}^{(k)}]\ket{j}\ket{\mbox{junk}}.
\end{align}

Recall the definition of $\ket{\check{\phi}}$ (eq.~(\ref{def_phises})), and let $\ket{\tilde{\phi}}$ be the normalized vector
\begin{equation}
\ket{\tilde{\phi}}:=\sum_j\braket{j}{\check{\phi}}\ket{\psi^{(k)}_j}.
\label{def_tilde_phi}
\end{equation}
Then,
\begin{align}
&1=\bra{\check{\phi}}\sum_{b,y}\mu(b|y)\bar{M}^{(k)}_{b|y}(\CC_1)\ket{\check{\phi}}=\bra{\check{\phi}}\sum_{b,y}\mu(b|y)\bar{B}^{(k)}_{b|y}\ket{\check{\phi}}\nonumber\\
&=\sum_{j,l}\braket{\check{\phi}}{l}\bra{l}\left(\sum_{b,y}\mu(b|y)\bar{B}^{(k)}_{b|y}\right) \ket{j}\braket{j}{\check{\phi}}\nonumber\\
&\overset{(\ref{local_k_ops})}{=}\sum_{j,l}\braket{\check{\phi}}{l}\bra{\psi^{(k)}_l}\left(\sum_{b,y}\mu(b|y)B^{(k)}_{b|y}\right) \ket{\psi^{(k)}_j}\braket{j}{\check{\phi}}\nonumber\\
&=\bra{\tilde{\phi}}\sum_{b,y}\mu(b|y)B^{(k)}_{b|y}\ket{\tilde{\phi}}\nonumber\\
&=\bra{\tilde{\phi}}\sum_{b,y}\mu(b|y){\sf C}_{b|y}[A]\ket{\tilde{\phi}}.
\end{align}
Since the inequality 
\begin{equation}
\langle\sum_{b,y}\mu(b|y){\sf C}_{b|y}[A]\rangle\leq 1
\end{equation}
completely self-tests $\ket{\bar{\phi}}, \bar{A}$, there exist orthonormal states $\{\ket{\phi(l_1,...,l_p)}\}_{\vec{l}}$ and local isometries $W_1,...,W_p$ such that
\begin{equation}
\ket{\tilde{\phi}}=\sum_{l_1,...,l_p}\braket{l_1,...,l_p}{\bar{\phi}}\ket{\phi(l_1,...,l_p)},
\label{identity_phi_tilde}
\end{equation}
and
\begin{equation}
(W_1\otimes...\otimes W_p)f[A]\ket{\phi(l_1,...,l_p)}=f[\bar{A}]\ket{l_1,...,l_p}\ket{\mbox{junk}'},
\label{st_conn}
\end{equation}
for all tensor polynomials $f$.

Define $W:=W_1\otimes...\otimes W_p$ and note that, for any tensor polynomial $g$,
\begin{align}
&g[{\sf C}[A]]\ket{\tilde{\phi}}\nonumber\\
&=W^\dagger \sum_{\vec{l}}\ket{\vec{l}}\bra{\vec{l}}g[{\sf C}[\bar{A}]]\ket{\bar{\phi}}\ket{\mbox{junk}'}\nonumber\\
&=\sum_{\vec{l}}\bra{\vec{l}}g[{\sf C}[\bar{A}]]\ket{\bar{\phi}}\ket{\phi(l_1,...,l_p)}.
\end{align}
As $\ket{\bar{\phi}}\in I(\CC_1)$, it holds that (remember that $\bar{B}^{(k)}=\bar{M}(\CC_1)$)
\begin{align}
\bra{\vec{l}}g[{\sf C}[\bar{A}]]\ket{\bar{\phi}}=\bra{\vec{l}}V(\CC_1)^\dagger g[\bar{B}^{(k)}]\ket{\check{\phi}}.
\end{align}
Thus,
\begin{align}
&g[{\sf C}[A]]\ket{\tilde{\phi}}\nonumber\\
&=\sum_{\vec{l}}\bra{\vec{l}}V(\CC_1)^\dagger g[\bar{B}^{(k)}]\ket{\check{\phi}}\ket{\phi(\vec{l})}.
\label{piece_1}
\end{align}

On the other hand, 
\begin{align}
&g[{\sf C}[A]]\ket{\tilde{\phi}}\nonumber\\
&=g[B^{(k)}]\ket{\tilde{\phi}}\nonumber\\
&\overset{(\ref{local_k_ops})}{=}\sum_{i}\bra{i}g[\bar{B}^{(k)}]\ket{\check{\phi}}\ket{\psi^{(k)}_i}.
\label{piece_2}
\end{align}
Equaling the two ends of eqs. (\ref{piece_1}), (\ref{piece_2}), we arrive at
\begin{align}
&\sum_{\vec{l}}\bra{\vec{l}}V(\CC_1)^\dagger g[\bar{B}^{(k)}]\ket{\check{\phi}}\ket{\phi(\vec{l})}\nonumber\\
&=\sum_{i}\bra{i}g[\bar{B}^{(k)}]\ket{\check{\phi}}\ket{\psi^{(k)}_i}.
\end{align}

Due to eq. (\ref{cyclicity}) and the relation $V(\CC_1)I(\CC_1)=H(\CC_1)=H_k(\CC_2)$, we have that $\{g(\bar{B}^{(k)})\ket{\check{\phi}}:g\}=H(\CC_1)$. Thus, for any vector $\ket{\beta}\in H(\CC_1)$, there exists a tensor polynomial $g_\beta$ such that $g_\beta[\bar{B}^{(k)}]\ket{\check{\phi}}=\ket{\beta}$. Choosing such tensor polynomials for the vectors $\{\ket{r}:r=1,...,\mbox{dim}(H(\CC_1))\}$, we have that, for any $r\in \{1,...,\mbox{dim}(H(\CC_1))\}$, 
\begin{align}
&\sum_{\vec{l}}\bra{\vec{l}}V(\CC_1)^\dagger \ket{r}\ket{\phi(\vec{l})}\nonumber\\
&=\sum_{i}\braket{i}{r}\ket{\psi^{(k)}_i}\nonumber\\
&=\ket{\psi^{(k)}_r}.
\label{def_new_vect}
\end{align}

Now, let us define the vectors:
\begin{align}
\ket{\psi(i_1,...,i_{k-1},\vec{l},i_{k+1},...,i_q)}:=(\bigotimes_{j<k}U_j\otimes \id_k\otimes \bigotimes_{j>k}U_j)^\dagger\ket{i_1,...,i_{k-1}}\ket{\phi(\vec{l})}\ket{i_{k+1},...,i_q}.
\label{def_new_vectors}
\end{align}
Then, by eqs. (\ref{state_B}) and (\ref{def_new_vect}) we have that
\begin{align}
&\ket{\psi}=\sum_{\vec{i}}\braket{\vec{i}}{\bar{\psi}(\CC_2)}\sum_{\vec{l}}\bra{\vec{l}}V(\CC_1)^\dagger \ket{i_k}\ket{\psi(i_1,...,i_{k-1},\vec{l},i_{k+1},...,i_q)}\nonumber\\
&=\sum_{i_1,...,i_{k-1},\vec{l},i_{k+1},...,i_q}\braket{i_1,...,i_{k-1},\vec{l},i_{k+1},...,i_q}{\bar{\psi}(\DD)}\ket{\psi(i_1,...,i_{k-1},\vec{l},i_{k+1},...,i_q)}.
\end{align}
Property (\ref{state_decomp}) thus holds for the inequality ${\sf D}=\mu(\CC_2)({\sf C}(\DD))$.

By eqs. (\ref{reconst_prop}), (\ref{identity_phi_tilde}), (\ref{st_conn}) we have that, for all $\vec{l}$, there exists a tensor polynomial $g_{\vec{l}}$ such that
\begin{equation}
g_{\vec{l}}[A]\ket{\tilde{\phi}}=\ket{\phi(\vec{l})}.
\label{decomp_phi_l}
\end{equation}
Denote by $i^0$ the indices $(i_1,...,i_{k-1},i_{k+1}, ...,i_{q})$. Then, we have that (for simplicity, we omit the tensor products in the following lines):
\begin{align}
&WU^0f[A]h[B^0]\ket{\psi(i^0,\vec{l})}\nonumber\\
&\overset{(\ref{def_new_vectors}),(\ref{decomp_phi_l})}{=}WU^0f[A]g_{\vec{l}}[A]h[B^0](U^0)^\dagger \ket{\tilde{\phi}}\ket{i^0}\nonumber\\
&\overset{(\ref{def_tilde_phi})}{=}Wf[A]g_{\vec{l}}[A]U^0h[B^0](U^0)^\dagger \sum_j\braket{j}{\check{\phi}}\ket{\psi^{(k)}_j}\ket{i^0}\nonumber\\
&\overset{(\ref{def_psi_k})}{=}Wf[A]g_{\vec{l}}[A]U^0h[B^0]U^\dagger \sum_j\braket{j}{\check{\phi}}\ket{j}\ket{\mbox{junk}}\ket{i^0}\nonumber\\
&=Wf[A]g_{\vec{l}}[A]U_k^\dagger U_kU^0h[B^0]U^\dagger \sum_j\braket{j}{\check{\phi}}\ket{j}\ket{\mbox{junk}}\ket{i^0}\nonumber\\
&\overset{(\ref{st_corresp_b})}{=}Wf[A]g_{\vec{l}}[A]U_k^\dagger  \sum_j\braket{j}{\check{\phi}}\ket{j}\ket{\mbox{junk}}h[\bar{B}^0]\ket{i^0}\nonumber\\
&\overset{(\ref{def_psi_k})}{=}Wf[A]g_{\vec{l}}[A]  \sum_j\braket{j}{\check{\phi}}\ket{\psi^{(k)}_j}h[\bar{B}^0]\ket{i^0}\nonumber\\
&\overset{(\ref{def_tilde_phi})}{=}Wf[A]g_{\vec{l}}[A]  \ket{\tilde{\phi}}h[\bar{B}^0]\ket{i^0}\nonumber\\
&\overset{(\ref{decomp_phi_l})}{=}Wf[A]\ket{\phi(\vec{l})}h[\bar{B}^0]\ket{i^0}\nonumber\\
&\overset{(\ref{st_conn})}{=}f[\bar{A}]\ket{\vec{l}}h[\bar{B}^0]\ket{i^0}\ket{\mbox{junk}'}.
\end{align}
\noindent Since the above holds for arbitrary tensor polynomials $f,g$, it also holds for linear combinations thereof, and so for general tensor polynomials of $\bar{D}=(\bar{B}^{(1)},...,\bar{B}^{(k-1)},\bar{A}^{(1)},...,\bar{A}^{(p)},\bar{B}^{(k+1)},...,\bar{B}^{(q)})$. Thus Property (\ref{st_corresp}) also holds for ${\sf D}$, which makes this Bell functional fully self-test $(\ket{\bar{\psi}(\DD)},\bar{A}(\DD))$.

The complex $\DD$ has been proven fully self-testing.
\end{proof}

\section{Connector complexes supplementary information}
\label{sec:other_complexes}

In this Appendix, we provide additional details over the connector complexes introduced in Sec.~\ref{sec:firstcomplexes} of the main text. These consist of couples nontrivial Bell functionals $\sB_{0,1}$ that have the following properties:
\begin{enumerate}
    \item When applied on dichotomic $\pm 1$ quantum operators $K$, the functionals are normalized as $|\sB_{0,1}(K)|\leq 1$;
    \item There exist a $2$-dimensional subspace $\mathbb{S}={\sf Span}\{\ket{\phi_+},\ket{\phi_-}\}$ and dichotomic operators $\bar{K}$ such that both $\sB_{0}[\bar{K}]$ and $\sB_{1}[\bar{K}]$ reach their extremal values $\pm 1$ within $\mathbb{S}$.
\end{enumerate}

Bell functionals that satisfy the two properties above can be used to define tight connector complexes via the scheme described in the main text Sec.~\ref{sec:firstcomplexes}, i.e. namely by Proposition~\ref{prop:pauli_construction}.

Additionally, all the proposed Bell functionals are fully self-testing (Def.~\ref{def:FST}). In order to show this, we will exploit Proposition~\ref{prop:dicho_FST} (proven in App.~\ref{sec:qubit_selftest_lemmas}), which only requires two simple verifiable conditions to be satisfied. We will verify such conditions separately for each class of Bell functionals in the following subsections. This will show that the connectors proposed in the main text Sec.~\ref{sec:firstcomplexes} satisfy the first part~\eqref{eq:fst_bell_exist} of the full self-testing definition. The second part~\eqref{cyclicity} is straightforward to verify in general, as all presented complexes are based on the Pauli construction outlined above and in Prop.~\ref{prop:pauli_construction}, for which $I(\CC)\equiv\mathbb{S}$, with $\ket{\bar{\psi}(\CC)}=\ket{\phi_+}$ being, w.l.o.g., the $+1$ eigenvectors of $\sC_{+1|1}=\frac{1+\sB_1[\bar{K}]}{2}$, while by construction $\bra{\phi_-}\sC_{b|0}\ket{\phi_+}\neq 0$.

\subsection{Tilted CHSH}
\label{secapp:tilted_CHSH}
The Tilted CHSH connector complex $\mathcal{C}_{\rm Til}$ is constructed starting from the Tilted CHSH operator functionals~\cite{acin2012randomness,bamps2015sum}
\begin{align}
    \beta(\theta) {\sf B}_{0,\theta}[K] &= \alpha(\theta) K^{(1)}_1\otimes \id^{(2)}  -K^{(1)}_0\otimes K^{(2)}_0 + K^{(1)}_1 \otimes K^{(2)}_0 + K^{(1)}_0 \otimes K^{(2)}_1 + K^{(1)}_1 \otimes K^{(2)}_{1}\;,\\
    \beta(\theta) {\sf B}_{1,\theta}[K] &= \alpha(\theta) K^{(1)}_0\otimes \id^{(2)} + K^{(1)}_0\otimes K^{(2)}_0 + K^{(1)}_1 \otimes K^{(2)}_0 + K^{(1)}_0 \otimes K^{(2)}_1 - K^{(1)}_1 \otimes K^{(2)}_1\;,
\end{align}
where
\begin{align}  \alpha(\theta)=2\left(1+2\tan(2\theta)^{2}\right)^{-\frac{1}{2}}\;.
\end{align}
The quantum bound $\beta(\theta)$  of the corresponding tilted CHSH inequality is given by~\cite{bamps2015sum} 
\begin{align}
    \beta(\theta)=\sqrt{8+2\alpha^{2}(\theta)}
\end{align}
so  that $-1\leq \sB_{0,1}\leq 1\;$ for any choice of dichotomic operators $K$. 

The saturation of the quantum Bell inequality $\sB_{1,\theta}$ is obtained by measuring the bipartite non-maximally entangled states
\begin{align}
    \ket{\phi_+} =\cos\theta\ket{00}+\sin\theta\ket{11}\;,\quad
    \ket{\phi_-} =\sin\theta\ket{01}-\cos\theta\ket{10}\;,
\end{align}
with the operators
\begin{align}
    \bar{K}^{(1)}_0 &=\sigma_3\;,\; &\bar{K}^{(1)}_1&=\sigma_1\;,\\
    \bar{K}^{(2)}_0 &=\cos\mu(\theta)\, \sigma_3 + \sin\mu(\theta)\, \sigma_1\;,\; &\bar{K}^{(2)}_1&=\cos\mu(\theta)\, \sigma_3 - \sin\mu(\theta)\, \sigma_1\;,
\end{align}
with $\mu(\theta)=\arctan{(\sin{(2\theta)})}$.
With this choice, one has
\begin{align}
    \sB_1[\bar{K}]\ket{\phi_\pm}=\pm\ket{\phi_\pm}\;,\quad
    \sB_0[\bar{K}]\ket{\phi_\pm}=\ket{\phi_\mp}\;,
\end{align}
from which tight connector complexes $\CC_{\rm Tilt}$ can be built through~\eqref{eq:first_tight_connectors}.

\paragraph{The Tilted CHSH is fully self-testing.}
\label{subsec:Tilted_is_FST}
In~\cite{bamps2015sum}, it is proven that the saturation of the Tilted CHSH $\sB_1$ self-tests -- in the standard sense of Def.~\ref{def:standard_self_testing} -- the optimal operators $\bar{K}$ and states $\ket{\phi_{\pm}}$ (clearly, $\sB_0$ enjoys a similar property and can be used to selftest $\frac{\ket{\phi_+}\pm\ket{\phi_-}}{\sqrt{2}}$).

Starting from the standard self-testing property, we can prove that $\sB_1$ \emph{fully} self-tests (Def.~\ref{def:FST}) via Proposition~\ref{prop:dicho_FST}. Namely, one can verify that:

\emph{i)} when decomposing $\sB_1$ as in~\eqref{eq:corform}, the operators ${\sf D}^{(k)}_x[\bar{K}]$ are invertible (for $x\neq\star$, or equivalently $K_x^{(k)}\neq\id^{(k)}$). In the case of the Tilted CHSH above, it is straightforward to notice that, as each ${\sf D}^{(k)}_x[\bar{K}]$ involves linear combinations of non-parallel Pauli measurements, this condition is always satisfied for both $\sB_0$ and $\sB_1$ when $\theta$ is nontrivial $\theta\neq m\pi/2 $, $m\in\mathbb{N}$.

%\emph{ii)} for some choice $\vec{x}^*$ if inputs, all possible outputs have nonzero probability $P(\vec{a}|\vec{x}^*)$. Again, it is straightforward to verify that for $\theta\neq m\pi/2$, one can choose e.g. $\{x^*_1,x^*_2\}=\{0,0\}$ and notice that, as $\ket{\psi_+}$ is of the form $\alpha\ket{00}+\beta\ket{11}$ all four outcomes $\{a_1,a_2\}=\{\pm 1,\pm 1\}$ have nonzero probability, as $K^{(1)}_0=\sigma_3$ (the $Z$ Pauli matrix) and
%\begin{align}
    %\bra{\phi_+}\frac{\id^{(1)}\pm\sigma_3^{(1)}}{2}\otimes \frac{\id^{(2)}\pm\bar{K}_0^{(2)}}{2}\ket{\phi_+}\neq 0,
%\end{align}
%due to the eigenvectors of $\id\pm\sigma_3$ being $\ket{0}$ and $\ket{1}$, and $\bar{K}_0^{(2)}$ being nondiagonal in such basis. Clearly, the same argument holds for $\ket{\psi_-}$.

\emph{ii)} For $k=1$, and each $x=0,1$, and, w.l.o.g., $\ket{\psi}\equiv\ket{\phi_+}$,
    \begin{equation}
    {\sf{Span}}\{\bar{A}_{a|x}^{(1)}\otimes \bar{A}_{a'|x'}^{(2)}\ket{\bar{\psi}}:a,a',x'\}=H_1 \otimes H_2.
    \end{equation}
This can be verified explicitly, noticing for example that for all choices of $\{x,x'\}$ and $\theta\neq m\pi/2$, one has $P(a,a'|x,x')\neq 0$ $\forall\{a,a'\}$, meaning that $\bar{A}_{a|x}^{(1)}\otimes \bar{A}_{a'|x'}^{(2)}\ket{\bar{\psi}}\neq 0$, and clearly spanning  $H_1 \otimes H_2$ when varying $\{a,a'\}$. The same argument applies when choosing $k=2$.

Thus, by Proposition~\ref{prop:dicho_FST}, the Tilted CHSH inequalities are fully self-testing.

\subsection{Wooltorton-Brown-Colbeck}
The authors (WBC) of~\cite{wooltorton2023device} define the following family of Bell functionals:
\begin{align}
  {\beta(\theta,\varphi,\omega)}\sB_{1}(\theta,\varphi,\omega)[K] := &\cos(\theta + \varphi)\cos(\theta + \omega)\big(\cos(\omega) K^{(1)}_0 \otimes K^{(2)}_0  - \cos(\varphi) K^{(1)}_0 \otimes K^{(2)}_{1}\big) + \nonumber\\ 
   &\cos(\varphi)\cos(\omega)\big(\cos(\theta + \varphi) K^{(1)}_{1}\otimes K^{(2)}_{1} - \cos(\theta + \omega) K^{(1)}_{1}\otimes K^{(2)}_{0} \big) ,
\end{align}
with $\beta(\theta,\varphi,\omega):=\sin(\theta)\sin(\omega - \varphi)\sin(\theta + \varphi + \omega)$.
If the angles $\theta, \phi, \omega \in \mathbb{R}$ satisfy the condition
\begin{equation}
    \cos(\theta + \phi)\cos(\phi)\cos(\theta + \omega)\cos(\omega) < 0,
    \label{class_cond}
\end{equation}
then the maximum (minimum) quantum value of the Bell functional $\sB_1$ is $1$ ($-1$). Moreover, the maximal values $\pm 1$ of  $\sB_{1}(\theta,\phi,\omega)$ self-test the state and operators
\begin{align}
\ket{\phi_{\pm}} &=\frac{1}{\sqrt{2}}(\ket{00}\pm i\ket{11})\;, & & \nonumber\\
\bar{K}^{(1)}_0 &=\sigma_1\;, & &\bar{K}^{(1)}_1=\cos(\theta)\sigma_1+\sin(\theta)\sigma_2\;,\nonumber\\
\bar{K}^{(2)}_0 &=\cos(\varphi)\sigma_1+\sin(\varphi)\sigma_2\;,& &  \bar{K}^{(2)}_1=\cos(\omega)\sigma_1+\sin(\omega)\sigma_2\;.
\label{ref_W}
\end{align}
Consider now the complementary functional 
\begin{equation}
{\sB}_{0}(\theta,\varphi,\omega)[K^{(1)}_0,K^{(1)}_1,K^{(2)}_0,K^{(2)}_1]:=\sB_{1}(\theta,\varphi,\omega)[K^{(1)}_1,K^{(1)}_0,K^{(2)}_1,K^{(2)}_0].
\end{equation}
This is obviously another normalized quantum Bell inequality. One can verify that not only $\sB_{0}(\theta,\varphi,\omega)[\bar{K}]\ket{\phi_\pm}=\pm \ket{\phi_\pm}$, but also
\begin{align}
    \sB_{0}(\theta,\varphi,\omega)[\bar{K}]\ket{\phi_+}& =-\cos(\theta+\varphi+\omega)\ket{\phi_+}+i\sin(\theta+\varphi+\omega)\ket{\phi_-}\;\\
    \sB_{0}(\theta,\varphi,\omega)[\bar{K}]\ket{\phi_-}& =\cos
    (\theta+\varphi+\omega)\ket{\phi_-}-i\sin(\theta+\varphi+\omega)\ket{\phi_+}\;.
\end{align}
These relations show that one can build, via Proposition~\ref{prop:pauli_construction} a tight connector complex $\mathcal{C}_{\rm WBC}$ starting from the WBC inequalities. More specifically by defining the co-isometry
\begin{align}
    V=\ket{0}\bra{\phi_-}+\ket{1}\bra{\phi_+}
\end{align}
we see that
\begin{align}
    V^\dagger\sB_1[\bar{K}]V=\sigma_3\;, 
    \quad\text{and}\quad
    V^\dagger\sB_0[\bar{K}]V=\cos(\theta+\varphi+\omega+\pi)\sigma_3+\sin(\theta+\varphi+\omega+\pi)\sigma_2\;. 
\end{align}

\subsubsection{The Wooltorton-Brown-Colbeck connector is fully self-testing}
\label{subsec:WBC_selftest}
In order to prove that the WBC inequalities are fully self-testing, we make use once again of Proposition~\ref{prop:dicho_FST}, and similarly to the case of Tilted CHSH~\ref{subsec:Tilted_is_FST}, we simply notice that:

\emph{i)} when decomposing $\sB_1$ as in~\eqref{eq:corform}, the operators ${\sf D}^{(k)}_x[\bar{K}]$ are invertible. In fact, each ${\sf D}^{(k)}_x[\bar{K}]$ involves linear combinations of non-parallel Pauli measurements, unless the choice of angles is trivial, e.g. $\theta=m\pi$ with $m\in\mathbb{N}$ (making $\bar{K}^{(1)}_0=\pm \bar{K}^{(1)}_1$), or $\omega=\varphi+m\pi$ with $m\in\mathbb{N}$ (making $\bar{K}^{(2)}_0=\pm \bar{K}^{(2)}_1$).

%\emph{ii)} secondly, for some choice $\vec{x}^*$ if inputs, all possible outputs have nonzero probability $P(\vec{a}|\vec{x}^*)$. Consider first $\varphi\neq \pi/2 +m\pi$ and $\{x^*_1,x^*_2\}=\{0,0\}$.
%The eigenvectors of $\bar{K}_0^{(1)}=\sigma_1$ are $\ket{\pm}=\frac{\ket{0}\pm\ket{1}}{\sqrt{2}}$. One can then simply verify
%\begin{align}
%    \bra{\pm}_1\ket{\phi_+}_{12}=\frac{1}{2}\left(\ket{0}_2\pm i\ket{1}_2\right)\;,
%\end{align}
%from which it follows that both outputs of the measurement on the second party have nonzero probability, as far as $\bar{K}_0^{(2)}\neq \sigma_2$, which is guaranteed if
%$\varphi\neq \pi/2+m\pi$.
%Finally, notice that in case $\varphi=\pi/2+m\pi$ we can repeat the argument considering $\bar{K}^{(2)}_1$, that is $\{x^*_1,x^*_2\}=\{0,1\}$, given that we assume nontrivial choices of angles $\omega\neq\varphi+m\pi$ as in the previous point \emph{i)}.

\emph{ii)} For $k=1$, and each $x=0,1$, and, w.l.o.g., $\ket{\psi}\equiv\ket{\phi_+}$,
    \begin{equation}
    {\sf{Span}}\{\bar{A}_{a|x}^{(1)}\otimes \bar{A}_{a'|x'}^{(2)}\ket{\bar{\psi}}:a,a',x'\}=H_1 \otimes H_2.
    \end{equation}
This can be verified  again explicitly using a similar argument to the one used for the Tilted CHSH above, i.e. noticing for that for all choices of $\{x,x'\}$ and for nontrivial angles ($\theta\neq m\pi$ and $\omega=\varphi+m\pi$), for each fixed $x$, at least one between $K^{(2)}_0$ and $K^{(2)}_1$ satisfies $K^{(2)}_{x'}\neq \pm K^{(1)}_x$. With such choice one has $P(a,a'|x,x')\neq 0$ $\forall\{a,a'\}$, meaning that $\bar{A}_{a|x}^{(1)}\otimes \bar{A}_{a'|x'}^{(2)}\ket{\bar{\psi}}\neq 0$, and clearly spanning  $H_1 \otimes H_2$ when varying $\{a,a'\}$. The same argument applies when choosing $k=2$.

Thus, via Proposition~\ref{prop:dicho_FST}, the WBC inequalities are fully self-testing.

\subsection{Graph states Bell inequalities  and $q\rightarrow 1$ tight connectors for $q>2$}
\label{app:graph}
\emph{Nota bene:} in this appendix we will use for simplicity of exposition the notation
\begin{align}
    \{\sigma_1,\sigma_2,\sigma_3\}=:\{X,Y,Z\}
\end{align}
denoting the 3 Pauli matrices. Moreover, we will indicate the party/vertex on which these operators act with a simple subscript $X_i$ instead of superscript $X^{(i)}$. Moreover, we will use the symbol $N$ to indicate the number of parties associated to the connector, i.e.
\begin{align}
    q=N \;.
\end{align}
\vspace{0.5cm}

Thus, we now construct Bell functionals that can be used to build $N\rightarrow 1$ tight connectors, inspired by the BASTA Bell inequalities presented in~\cite{baccari2020scalable}.
To start, we remind, for a given graph $\G:=\{\V,\E\}$ -- specified by a set of vertices of cardinality $|\V|=q$ and a set of edges $\E$ -- the definition of the corresponding graph state \cite{hein2006entanglement}.
First, for each vertex $v_i$ of the graph we define the \emph{stabilizer operator}
\begin{align}
    G_i:=X_i\otimes\bigotimes_{j\in n(i)}Z_j \;,
\end{align}
where $n(i)$ indicates the set of neighbours of $i$, i.e. $n(i):\{j:\{i,j\}\in\mathcal{E}\}$. Stabilizers commute among themselves and can be used to provide a basis of the $2^N$ dimensional Hilbert space of the $N$-qubits of the graph. The graph ($\ket{\phi_+}$) and antigraph ($\ket{\phi_-}$) states of $\G$ are, precisely, the only normalized vectors (modulo phases) satisfying:
\begin{align}
    G_i\ket{\phi_\pm}=\pm\ket{\phi_\pm}\; \forall i\;.
\end{align}
An explicit expression for these states can be provided as
\begin{align}
\label{eqapp:phi+_graph}
    \ket{\phi_+} &=\bigotimes_{ij\in\E} CZ_{ij} \ket{+}^{N}\;,\\
\label{eqapp:phi-_graph}
    \ket{\phi_-} &=\bigotimes_{i\in \V} Z_i \bigotimes_{ij\in\E} CZ_{ij} \ket{+}^{N}\;,
\end{align}
where $CZ_{ij}$ ($\equiv {\rm diag}(1,1,1,-1)$ in the computational basis)
is the control-$Z$ operator on qubits $\{i,j\}$. More explicitly in the computational basis one has
\begin{align}
\label{eqapp:phi+_graph_components}
    \ket{\phi_+} &=2^{-N/2}\sum_{\vec{x}} (-1)^{\sum_{ij\in\E}x_i x_j} \ket{\vec{x}}\;,\\
\label{eqapp:phi-_graph_components}
    \ket{\phi_-} &=2^{-N/2}\sum_{\vec{x}} (-1)^{\sum_{ij\in\E}x_i x_j + \sum_{i\in \V}x_i} \ket{\vec{x}}\;.
\end{align}

Finally, we provide an observation that will be useful later.
Define the stabilizer operator $\tilde{G}_i^{(l)}$ as the modified version of $G_i$ under the swapping $X_l\leftrightarrow Z_l$, that is
\begin{align}
\label{eqapp:tilde_stabilizer}
    \tilde{G}_i^{(l)}:=
    \begin{cases}
     Z_i\bigotimes_{j\in n(i)}Z_j & l=i\;,\\
     X_i\otimes X_l \otimes \bigotimes_{l\neq j\in n(i)}Z_j & l\in n(i)\;,\\
     G_i & l\notin \{i\}\cup n(i)\;.
    \end{cases}
\end{align}
%For example, $\tilde{G}_i^{(i)}=Z_i\otimes\bigotimes_{j\in n(i)}Z_j$, or if $l\in n(i)$, $\tilde{G}_i^{(l)}=X_i\otimes X_l \otimes \bigotimes_{l\neq j\in n(i)}Z_j$.
The observation is the following: whenever $l$ is a neighbour of both $i$ and $i'$, the anticommutator
\begin{align}
\label{eqapp:obs_anticom}
    \{\tilde{G}_i^{(l)},G_{i'}\}=0
\end{align}
is null.

Consider now the BASTA Bell inequality introduced in~\cite{baccari2020scalable}. Given a vertex $v_i$ with $n_i$ neighbours it reads
\begin{align}
    \sB(\mathcal{G},v_i)=n_i (K_0^{(i)}+K_1^{(i)})\bigotimes_{j\in n(i)} K_1^{(j)} + \sum_{j\in n(i)}  (K_0^{(i)}-K_1^{(i)})\otimes K_0^{(j)}\bigotimes_{i\neq k\in n(j)} K_1^{(k)} %+\\&
    + \sum_{j \notin n(i)\cup \{i\}} K_0^{(j)}\bigotimes_{k\in n(j)}K^{(k)}_1 \;,
    \label{eqapp:BASTA_vi}
\end{align}
This Bell functional satisfies~\cite{baccari2020scalable}
\begin{align}
|\sB(\mathcal{G},v_i)|\leq (2\sqrt{2}-1) n_i + N-1\;.
\end{align}
Moreover, the saturation of the upper bound $\sB(\mathcal{G},v_i)=(2\sqrt{2}-1) n_i + N-1$ self-tests (in the standard sense, see Def.~\ref{def:standard_self_testing}) the associated graph-state
\begin{align}
    \ket{\psi}=\ket{\phi_+(\mathcal{G})}
\end{align}
and the operators
\begin{align}
    \bar{K}_0^{(i)} &=\frac{X_i+Z_i}{\sqrt{2}},\;& \bar{K}_1^{(i)} &=\frac{X_i-Z_i}{\sqrt{2}},\;\\
    \bar{K}_0^{(j)} &=X_j,\;& \bar{K}_1^{(j)} &= Z_j\qquad \qquad j\neq i\;.
\end{align}
In fact, it is easy to verify that with such a choice all the terms in~\eqref{eqapp:BASTA_vi} become proportional to corresponding stabilizer operators of $\mathcal{G}$, from which the value $(2\sqrt{2}-1) n_i + N-1$ stems. Consequently, it is easy to verify that the opposite value $-(2\sqrt{2}-1) n_i - N+1$ is obtained with the same measurements acting on $\ket{\phi_-(\mathcal{G})}$. This ensures that the renormalized version presented in the main text~\eqref{eq:BASTA} effectively acts as $\sigma_3\equiv Z$ operator on the subspace ${\sf Span}(\ket{\phi_+(\mathcal{G})},\ket{\phi_-(\mathcal{G})})$.

In order to use these properties as building blocks of a tight connector as from Prop.~\ref{prop:pauli_construction}, we next construct a dual Bell functional $\sB'$ such that $\sB'[\bar{K}]\ket{\phi_{\pm}(\mathcal{G})}\propto\ket{\phi_\mp(\mathcal{G})}$.

We present the construction of such $\sB'$ for the paradimatic examples of a Star-shaped graph and Triangle graphs in the next paragraphs, before proceeding to the case of a generic graph.

(\emph{Note:} the renormalized versions of $\sB$ and $\sB'$ are denoted, respectively, $\sB_1$ and $\sB_0$ in the main text)
\vspace{0.5cm}

\paragraph*{Star graph states.}\ \\
Consider a graph with a central vertex $v_1$ connected to all other $N-1$ vertices. In such case~\eqref{eqapp:BASTA_vi} becomes
\begin{align}
    \sB(\mathcal{G}_\ast,v_1)[K]=(N-1) (K_0^{(1)}+K_1^{(1)})\bigotimes_{j\geq 2} K_1^{(j)} + \sum_{j\geq 2}  (K_0^{(1)}-K_1^{(1)})\otimes K_0^{(j)}\;.
\end{align}
Similarly to the Tsirelson connector case one can verify that
\begin{align}
    \sB'(\mathcal{G}_\ast,v_1)[K]=(N-1) (K_0^{(1)}-K_1^{(1)})\bigotimes_{j\geq 2} K_1^{(j)} - \sum_{j\geq 2}  (K_0^{(1)}+K_1^{(1)})\otimes K_0^{(j)}
\end{align}
satisfies the desired properties.

Indeed, when substituting the optimal operators $\bar{K}$ in $\sB$ one obtains (we omit the labels $(\mathcal{G}_\ast,v_1)$ in the following)
\begin{align}
    \sB[{\bar{K}}]=\sqrt{2}(N-1)G_1+\sqrt{2}\sum_{i\geq 2} G_i
\end{align}
from which clearly
\begin{align}
    \sB[\bar{K}]\ket{\phi_{\pm}}=\pm2\sqrt{2}(N-1)\ket{\phi_\pm}\;.
\end{align}
More over, one can also verify that for the star graph $\mathcal{G}_\ast$ it holds
\begin{align}
    \bigotimes_{i\in V} Z_i \ket{\phi_\pm}=\ket{\phi_\mp}
\end{align}
as well as
\begin{align}
    X_1\otimes X_i \ket{\phi_\pm}=-\ket{\phi_\mp} \quad \forall i\geq 2\;,
\end{align}
directly from the expression
\begin{align}
    \ket{\phi_\pm(\mathcal{G_\ast})}=\frac{\ket{0\pm\pm\pm\dots}\pm\ket{1\mp\mp\mp\dots}}{\sqrt{2}}\;.
\end{align}
This in particular implies that
\begin{align}
\sB'[\bar{K}]= \sqrt{2}(N-1)\bigotimes_{i\in V} Z_i \ket{\phi_\pm}- \sqrt{2}\sum_{i\geq 2}X_1\otimes X_i\\
\Rightarrow    \sB'[\bar{K}]\ket{\phi_{\pm}}=2\sqrt{2}(N-1)\ket{\phi_{\mp}}\;.
\end{align}

\paragraph*{Triangle graph states.}\ \\
A similar machinery can be employed for a triangular graph $\mathcal{G}_\triangle$, by using the fact that
\begin{align}
    \ket{\phi_\pm(\mathcal{G}_\triangle)}=\frac{1}{2\sqrt{2}}\left(\ket{000}\pm\ket{001}\pm\ket{010}\pm\ket{100}-\ket{011}-\ket{110}-\ket{101}\mp\ket{111}\right)
\end{align}
and the stabilizer operators $G_i$ satisfy for such graph
\begin{align}
    X_1\otimes Z_2 \otimes Z_3 \ket{\phi_\pm} &=\pm \ket{\phi_\pm}\;,\\
    Z_1\otimes X_2 \otimes Z_3 \ket{\phi_\pm} &=\pm \ket{\phi_\pm}\;,\\
    Z_1\otimes Z_2 \otimes X_3 \ket{\phi_\pm} &=\pm \ket{\phi_\pm}\;,
\end{align}
while the modified stabilizers $\tilde{G}_1^{(i)}$ (cf.~\eqref{eqapp:tilde_stabilizer})
\begin{align}
    Z_1\otimes Z_2 \otimes Z_3 \ket{\phi_\pm} &=\ket{\phi_\mp}\;,\\
    X_1\otimes X_2 \otimes Z_3 \ket{\phi_\pm} &= -\ket{\phi_\mp}\;,\\
    X_1\otimes Z_2 \otimes X_3 \ket{\phi_\pm} &= -\ket{\phi_\mp}\;.
\end{align}
It follows that for any choice of vertex $v_i$ the Bell functionals
\begin{align}
    \sB(\mathcal{G}_\triangle,v_i)[K]=2 (K_0^{(i)}+K_1^{(i)})\bigotimes_{j\neq i} K_1^{(j)} + \sum_{j\neq i}  (K_0^{(i)}-K_1^{(i)})\otimes K_0^{(j)}\;,
\end{align}
\begin{align}
    \sB'(\mathcal{G}_\triangle,v_i)[K]=2 (K_0^{(i)}-K_1^{(i)})\bigotimes_{j\neq i} K_1^{(j)} - \sum_{j\neq i}  (K_0^{(i)}+K_1^{(i)})\otimes K_0^{(j)},
\end{align}
Satisfy $\sB[\bar{K}]\ket{\phi_\pm}=\pm\ket{\phi_\pm}$ and $\sB'[\bar{K}]\ket{\phi_\pm}=\ket{\phi_\mp}$.

%\paragraph{the most general graph?}
%We now define the following class of graphs: a graph $\G:=(\V,\E)$ defined by its set of vertices $\V$ and edges $\E$ is said to be \emph{star-partitionable} if it exists a subset $\V'\subset \V$ of vertices such that each vertex is a neighbour of \emph{exactly one} element of $\V'$, or is itself in $\V'$. That is
%\begin{align}
%    \forall v_i\in(\V\setminus\V')\;, \;\exists !\; v'_i \in \V'\; s.t.\; \{v_i,v'_i\}\in\E\;.
%\end{align}
%
%\begin{figure}
%    \centering
%    \includegraphics[width=0.2\linewidth]{Small_Network.png}
%    \caption{Figure placeholder for Star-partitionable graphs}
%    \label{fig:enter-label}
%\end{figure}
%
%\begin{align}
%    \sB_0=\frac{\sum_{i\in\V'} \sb_0^{(j)}}{2\sqrt{2}\sum_{i\in\V'}n_i}\\
%    \sB_1=\frac{\sum_{i\in\V'} \sb_1^{(j)}}{2\sqrt{2}\sum_{i\in\V'}n_i}
%\end{align}

\paragraph*{General graph states.}\ \\
\label{par:general_graph_B'}
We can now generalize the previous construction to a general graph. We start from the general BASTA inequality~\eqref{eqapp:BASTA_vi} for any graph, recognizing that it satisfies $\sB[\bar{K}]\ket{\phi_\pm}\propto \pm\ket{\phi_\pm}$ as it is constructed in order to consist in a sum stabilizer operators $G_i$ when choosing the optimal measurements $\bar{K}$. As mentioned above, the $N$ stabilizers commute among themselves and can be used to provide a basis of the $2^N$ dimensional space of the $N$-qubits of the graph. In such basis, $\ket{\phi_+}$ might be written as $\ket{+1,+1,\dots,+1}$ as well as $\ket{\phi_-}=\ket{-1,-1,\dots,-1}$. 
Consider now an operator $O$ that anticommutes with the stabilizers, i.e.
\begin{align}
    \{O,G_i\}=0\;\quad\forall i\;.
\end{align}
It clearly follows that
\begin{align}
    G_i O\ket{\phi_\pm}=\mp O\ket{\phi_\pm}\quad \forall i\;.
\end{align}
That is, $O\ket{\phi_+}$ is proportional to $\ket{\phi_-}$ and viceversa. This gives us an intuition on what building blocks might be used to build a Bell functional $\sB'$ satisfying $\sB'[\bar{K}]\ket{\phi_\pm}=\ket{\phi_\mp}$.
One can in fact notice that each of the terms in both examples above $\sB'(\mathcal{G}_\ast,v_1)$ and $\sB'(\mathcal{G}_\triangle,v_i)$ is built such that the terms $K_0^{(i)}\pm K_1^{(i)}$ operating on the vertex $v_i$ defining $\{\sB,\sB'\}$ are swapped between the two functionals, and result in $\sqrt{2}X_i$ and $\sqrt{2}Z_i$ when choosing $\bar{K}$. This makes $\sB'[\bar{K}]$ to be formed by modified stabilizer operators $\tilde{G}_j^{(i)}$ in which the vertex $v_i$ has $X_i\leftrightarrow Z_i$ swapped. These $\tilde{G}_j^{(i)}$ anticommutes with all others $G_k$ because of the chosen $v_i$ always shares an edge with all vertices of the graph in the case of $\mathcal{G}_\ast$ and $\mathcal{G}_\triangle$ (cf. the observation~\eqref{eqapp:obs_anticom}). In more general graphs, there is not, in general, a vertex that neighbors with all the others.
In order to correct this issue, with some ingenuity we build an ansatz $\sB'[K]$ that adds tensor strings to the $\tilde{G}_j^{(i)}$ that form $\sB'[\bar{K}]$, such that each term in it has exactly one anticommuting ``party" with the original stabilizers.

Namely, we arrive at the ansatz:
\begin{align}
    \sB'(\G,v_i)[K] &= n_i (K_0^{(i)}-K_1^{(i)})\bigotimes_{j\neq i} K_1^{(j)} - \sum_{j\in n(i)}  (K_0^{(i)}+K_1^{(i)})\otimes K_0^{(j)}\bigotimes_{k\in {\sf EN}_{i,j}} K_1^{(k)}
     \label{eq:BASTA'}
\end{align}
where the Even Neighbours set ${\sf EN}_{i,j}$ is defined as those vertices that are either neighbours of both $v_i$ and $v_j$, or none of the two. One can indeed verify that all the addends in $\sB'[\bar{K}]$ anticommute with all stabilizers, and therefore swap between $\ket{\phi_+}$ and $\ket{\phi_-}$. 

The proportionality to the swap $\ketbra{\phi_+}{\phi_-}+\ketbra{\phi_-}{\phi_+}$ of each term is however insufficient to characterize the properties of the global Bell ansatz~\eqref{eq:BASTA'}. We thus verify them.
In the following, we choose for simplicity to label $v_i\equiv v_1$, and denote ${\sf EN}_{j}\equiv {\sf EN}_{1,j}$.
First, it should be noticed that as an operator $|\sB'|\leq n_1 2\sqrt{2}$. This can be verified directly by 
\begin{align}
     2\sqrt{2}(\left\langle\sB'\right\rangle+n_1 2\sqrt{2})&=n_1 \left\langle F_1^2\right\rangle + \sum_{i\in n(1)} \left\langle F_i^2\right\rangle\geq 0\;,\\
    F_1&=\sqrt{2} + (K_0^{(1)}-K_1^{(1)})\bigotimes_{i>1} K_1^{(i)} \;,\\
    F_i&=\sqrt{2} - (K_0^{(1)}+K_1^{(1)})\otimes K_0^{(i)}\bigotimes_{j\in {\sf EN}_i} K_1^{(j)}\;.
\end{align}
This proves $\langle\sB_1\rangle\geq -n_1 2\sqrt{2}$; one proves $\langle\sB_1\rangle\leq n_1 2\sqrt{2}$ similarly.

Finally we show that in the $\ket{\phi_\pm}$ subspace these limits are saturated when choosing the correlators $\bar{K}$. In particular, one can prove that
\begin{align}
    \sB'[\bar{K}]\ket{\phi_{\pm}}=n_1 2\sqrt{2}\ket{\phi_\mp}\;.
\end{align}
In fact
\begin{align}
   \sB'[\bar{K}]= n_1\sqrt{2}\; Z_1\bigotimes_{i\geq 2} Z_i -  \sqrt{2} \sum_{i\in n(1)}  X_1\otimes X_{i}\bigotimes_{j\in {\sf EN}_i} Z_{j}
\end{align}
Given that by definition~\eqref{eqapp:phi+_graph}-\eqref{eqapp:phi-_graph} $ Z_1\bigotimes_{i\geq 2} Z_i\ket{\phi_{\pm}}=\ket{\phi_\mp}$,
it only remains to verify that
\begin{align}
    -X_1\otimes X_{i}\bigotimes_{j\in {\sf EN}_i} Z_{j}\ket{\phi_{\pm}}=\ket{\phi_\mp}\;.
\end{align}
We verify this in the computational basis $x_i=0,1$ (we indicate by $\bar{x}_i=x_i\oplus 1$ the logical complement of $x_i$), in particular notice that 
\begin{align}
    & X_{i}\sum_{\vec{x}} (-1)^{\sum_{jk\in\E}x_j x_k} \ket{\vec{x}}\\
    =& X_{i}\sum_{\vec{x}}(-1)^{\sum_{j\in n(i)}x_ix_j+\sum_{jk\in\E,j\neq i,k\neq i}x_jx_k} \ket{\vec{x}}\\
    =& \sum_{\vec{x}}(-1)^{\sum_{j\in n(i)}\bar{x}_ix_j+\sum_{jk\in\E,j\neq i,k\neq i}x_jx_k} \ket{\vec{x}}\\
    =& \sum_{\vec{x}}(-1)^{\sum_{j\in n(i)}(\bar{x}_i-x_i)x_j+\sum_{j\in n(i)}x_ix_j+\sum_{jk\in\E,j\neq i,k\neq i}x_jx_k} \ket{\vec{x}}\\
    =& \sum_{\vec{x}}(-1)^{\sum_{j\in n(i)} x_j+\sum_{jk\in\E}x_jx_k} \ket{\vec{x}}\;.
\end{align}
The same steps can be used to find
\begin{align}
    &X_1 \sum_{\vec{x}}(-1)^{\sum_{j\in n(i)} x_j+\sum_{jk\in\E}x_jx_k} \ket{\vec{x}}\\
    =& (-1)^\delta \sum_{\vec{x}}(-1)^{\sum_{j\in n(i)} x_j+\sum_{j\in n(1)} x_j+\sum_{jk\in\E}x_jx_k} \ket{\vec{x}}
\end{align}
where $\delta=1$ if $v_1$ is a neighbour of $v_i$, and $\delta=0$ otherwise. As we are considering the former case, we can substitute $\delta=1$ to obtain
\begin{align}
    &-X_1\otimes X_{i}\ket{\psi^{+}}\\
    =& -\sum_{\vec{x}}(-1)^{\sum_{j\in n(i)} x_j+\sum_{j\in n(1)} x_j+\sum_{jk\in\E}x_jx_k} \ket{\vec{x}}\\
    =& -\sum_{\vec{x}}(-1)^{\sum_{j\notin {{\sf EN}_i} } x_j+\sum_{jk\in\E}x_jx_k} \ket{\vec{x}}
\end{align}
and finally
\begin{align}
    &-X_1\otimes X_{i}\bigotimes_{j\in {\sf EN}_i} Z_{j}\ket{\phi_+}\\
    =& -\sum_{\vec{x}}(-1)^{\sum_{j\notin {{\sf EN}_i} } x_j+\sum_{j\in {{\sf EN}_i} } x_j+\sum_{jk\in\E}x_jx_k} \ket{\vec{x}}\\
    =& -\sum_{\vec{x}}(-1)^{\sum_{j} x_j+\sum_{jk\in\E}x_jx_k} \ket{\vec{x}}\\
    \equiv & \ket{\phi_-}\; .
\end{align}
One can then notice that $\left(X^1\otimes X^{i}\bigotimes_{j\in {\sf EN}_i} Z^{j}\right)^2=\id$ in order to prove as a consequence the reverse relation 
\begin{align}
    -X_1\otimes X_{i}\bigotimes_{j\in {\sf EN}_i} Z_{j}\ket{\phi_-}=\ket{\phi_+}\;.
\end{align}

This concludes the verification of the desired properties os $\sB$ and $\sB'$ for general graphs $\mathcal{G}$.
\vspace{0.5cm}

\subsubsection{The BASTA Bell inequality is fully self-testing}
Once again we make use of Proposition~\ref{prop:dicho_FST} to prove that the inequality $\sB(\mathcal{G},v_i)$~\eqref{eqapp:BASTA_vi}(\eqref{eq:BASTA} in the main text), fully self-tests the graph state $\ket{\phi_+}$ (or $\ket{\phi_-}$ equivalently).

First, notice that condition 1. in Prop~\ref{prop:dicho_FST} is straightforward as the ${\sf{D}}^{(k)}_x[\bar{K}]$ operators consist, in this case, of strings of $X$, $Z$, and $\id$ operators, and are therefore invertible.

Condition 2. in Proposition~\ref{prop:dicho_FST} requires that for each vertex/party $v_k$, for fixed the measurement via $\bar{A}_{a|x_k}$ the span of vectors in the form $\bar{A}_{a|x}^{(k)}\otimes f[\{\bar{A}^{(j)}:j\not=k\}]\ket{{\phi_+}}$ covers the full Hilbert space.

We thus distinguish two cases:

\paragraph*{Case 1: $v_k=v_i$.}
Namely, if the chosen vertex on which to verify the condition 2. coincides with the vertex $v_i$ defining the BASTA inequality $\sB(\mathcal{G},v_i)$.
As the optimal $v_i$ measurement consists in $\bar{K}^{(i)}_0=\frac{X_i+Z_i}{\sqrt{2}}$ and $\bar{K}^{(i)}_1=\frac{X_i-Z_i}{\sqrt{2}}$, their eigenvalues are of the form (take $\bar{K}^{(i)}_0$ for simplicity) 
\begin{align}
\ket{h_{+1}} &=\cos(\pi/8)\ket{0}+\sin(\pi/8)\ket{1}\;,\\
\ket{h_{-1}} &=-\sin(\pi/8)\ket{0}+\cos(\pi/8)\ket{1} \;. 
\end{align}
For the rest of the vertices, consider now $\bar{K}^{(j\neq i)}_1$, that is, $Z_j$. This means that the tensor polynomials $f(\bar{A}^{(j\neq k)})$ in ~\eqref{gen_Hilbert_space} contain at least the projectors on all strings $\ket{\vec{x'}}_{v_j|j\neq k}$  $(\vec{x'}\in\{0,1\}^{N-1}$) of the computational basis. 
In particular, notice that
\begin{align}
    \bra{h_1}_{v_k}\bra{\vec{x'}}_{v_j|j\neq k} \ket{\phi_+} &\neq 0\quad \forall\vec{x'}\in\{0,1\}^{N-1}\\
    \bra{h_2}_{v_k}\bra{\vec{x'}}_{v_j|j\neq k} \ket{\phi_+} &\neq 0 \quad \forall\vec{x'}\in\{0,1\}^{N-1}\;.
\end{align}
This is due to the fact that $\ket{\phi_+}$~\eqref{eqapp:phi+_graph_components} has uniform amplitude (in absolute value) on all global strings vectors $\ket{\vec{x}}$ ($\vec{x}\in\{0,1\}^N$), together with the fact that $\cos(\pi/8)$ is strictly larger than $\sin(\pi/8)$.

As the eigenvalues of $\bar{K}_0^{(k)}$ and $\bar{K}_1^{(j\neq k)}$ generate the full Hilbert space, and similarly with $\bar{K}_1^{(k)}$, condition 2. of Prop.~\ref{prop:dicho_FST} is verified for $v_i=v_k$.

\paragraph*{Case 2: $v_k\neq v_i$.}
For $v_k\neq v_i$ two situations can happen. In case $x_k=1$, one has $\bar{K}^{(k)}_1=Z_k$, and by considering $\bar{K}^{(j\neq {i,k})}_1\equiv Z_{(j\neq {i,k})}$ together with $\frac{\bar{K}^{(i)}_0-\bar{K}^{(i)}_1}{\sqrt{2}}=Z_i$ the span 
\begin{align}
    {\sf Span} \{\bar{A}^{(k)}_{a_k|x_k=1}\otimes \bar{A}^{(i)}_{a_i|x_i} \bigotimes_{j\neq k,i}\bar{A}^{(j)}_{a_j|x_j=1}: \vec{a},x_i\}
\end{align}
can generate all projectors on the computational strings $\ket{\vec{x}}$, which generate the full Hilbert space and on which $\ket{\phi_+}$ has nonzero overlap.
In case $x_k=0$, one has that $\bar{K}^{(k)}_{x_k=0}=X_k$ and $\bar{A}^{(k)}_{a|x_k=0}$ is either the projector on $\ket{+}$ or on $\ket{-}$. Take the projection on $\ket{+}$. We have
\begin{align}
    \bra{+}_{v_k} \ket{\phi_+}& =2^{-N/2}\bra{+}_{v_k } \sum_{\vec{x}} (-1)^{\sum_{ij\in\E}x_i x_j} \ket{\vec{x}}\\
    &=2^{-(N+1)/2}\sum_{\vec{x'}}(1+(-1)^{\sum_{i\in n(k)}x'_i})(-1)^{\sum_{ij\in\E'}x'_i x'_j}\ket{\vec{x'}}\;,
    \label{eqapp:bra+_phi+_exp}
\end{align}
where $\E'$ is the graph obtained from $\E$ when removing $v_k$ and all edges attached to it. Such graph defines a new $N-1$-partite graph state $\ket{\phi'_+}$, and the above equation can be rewritten as
\begin{align}
    \bra{+}_{v_k} \ket{\phi_+}=\frac{1}{2}(1+\bigotimes_{i\in n(k)}Z_i)\ket{\phi'_+}\;.
    \label{eqapp:bra+_phi+_imp}
\end{align}
Similarly one obtains
\begin{align}
    \bra{-}_{v_k} \ket{\phi_+}=\frac{1}{2}(1-\bigotimes_{i\in n(k)}Z_i)\ket{\phi'_+}\;.
    \label{eqapp:bra-_phi+_imp}
\end{align}
Crucially, $\bra{-}_{v_k} \ket{\phi_+}$ is nonzero unless $n(k)$ is trivial (meaning that the original $\E$ is not fully connected).

It remains us to prove that degree-1 polynomials $f[\{\bar{A}^{(j)}:j\not=k\}]$ applied to $(1+\bigotimes_{i\in n(k)}Z_i)\ket{{\phi'_+}}$ can generate the full $N-1$-qubit Hilbert space ${\mathbb{C}^2}^{\otimes{N-1}}$, (and similarly for $(1-\bigotimes_{i\in n(k)}Z_i)\ket{{\phi'_+}}$). Such polynomials include the projections on all possible tensor-products of eigenvectors of $Z_i$ and $X_i$.
By direct inspection of~\eqref{eqapp:bra+_phi+_exp}-\eqref{eqapp:bra+_phi+_imp}, such vector contains strings vectors $\ket{\vec{x'}}$ satisfying $\sum_{i\in n(k)}x'_i=0 \mod{2}$.
By choosing $Z$-measurements on all the $N-1$ qubits, all such strings can be generated. 
Consider now a string $\vec{x'^*}$ that has instead an odd sum 
\begin{align}
    \sum_{i\in n(k)}x'^*_i=1 \mod{2}
\end{align}
and label $v_0$ one of the neighbouring vertices to $v_k$.
One clearly has
\begin{align}
    \bra{+}_{v_0}\bra{x'^*_1,x'^*_2,\dots}_{v_i\neq v_0}\;\ket{\vec{x'^{*}}}=\frac{1}{\sqrt{2}}\;,
\end{align}
given that $x'^*_1$ is either 0 or 1. Importantly, while
\begin{align}
    \bra{\vec{x'^{*}}}(1+\bigotimes_{i\in n(k)}Z_i)\ket{{\phi'_+}}=0\;,
\end{align}
when flipping the value of $x_0$
\begin{align}
    |\bra{\vec{x'^{*}}}X_0(1+\bigotimes_{i\in n(k)}Z_i)\ket{{\phi'_+}}|=\frac{2}{2^{(N-1)/2}}\neq 0\;.
\end{align}
It follows that
\begin{align}
    \bra{+}_{v_0}\bra{x'^*_1,x'^*_2,\dots}_{v_i\neq v_0}\;(1+\bigotimes_{i\in n(k)}Z_i)\ket{{\phi'_+}}\neq 0
\end{align}
is nonzero, and the corresponding projector applied to $(1+\bigotimes_{i\in n(k)}Z_i)\ket{{\phi'_+}}$ can be used to generate vectors proportional to
\begin{align}
    \ket{+}_{v_0}\ket{x'^*_1,x'^*_2,\dots}_{v_i\neq v_0}\;.
\end{align}
In linear combination with $X_0\ket{\vec{x'^*}}$, which can be generated by hypothesis using products of $Z$ measurements, we can therefore also generate the string vector $\ket{x'^*}$.

This concludes the argument:  using the projections on the $X_0$ and $Z_{i}$ eigenvectors, we can generate all strings vectors $\ket{\vec{x'}}$ from $(1+\bigotimes_{i\in n(k)}Z_i)\ket{{\phi'_+}}$; a symmetric reasoning applies to $(1-\bigotimes_{i\in n(k)}Z_i)\ket{{\phi'_+}}$,
and thus the span   
\begin{align}
    {\sf Span} \{\bar{A}^{(k)}_{a_k|x_k=0}\otimes \bar{A}^{(0)}_{a_0|x_0} \bigotimes_{j\neq k,0}\bar{A}^{(j)}_{a_j|x_j=1}: \vec{a},x_0\}
\end{align}
can generate the full $N$-qubit Hilbert space.

\end{appendix}

\end{document}